\documentclass[english,3p,twocolumn]{revtex4-1}

\usepackage[monochrome]{xcolor}
\usepackage[colorlinks]{hyperref}
\hypersetup{
    colorlinks=true,
    linkcolor=blue,
    filecolor=magenta,    
    urlcolor=cyan,
}

\usepackage[T1]{fontenc}
\usepackage[latin9]{inputenc}
\usepackage{amsmath}
\usepackage{amssymb, enumerate}
\usepackage{color}
\usepackage{babel}
\usepackage{graphicx}
\usepackage{epstopdf}
\usepackage{float}
\usepackage{hyperref}
\usepackage{enumerate}
\usepackage[normalem]{ulem}
\makeatletter

\@ifundefined{textcolor}{}
{%
 \definecolor{BLACK}{gray}{0}
 \definecolor{WHITE}{gray}{1}
 \definecolor{RED}{rgb}{1,0,0}
 \definecolor{GREEN}{rgb}{0,0.5,0}
 \definecolor{BLUE}{rgb}{0,0,1}
 \definecolor{CYAN}{cmyk}{1,0,0,0}
 \definecolor{MAGENTA}{cmyk}{0,1,0,0}
 \definecolor{YELLOW}{cmyk}{0,0,1,0}

}

\newtheorem{theorem}{Theorem}
\newtheorem{corollary}[theorem]{Corollary}
\newtheorem{proposition}[theorem]{Proposition}
\newtheorem{example}[theorem]{Example}
\newtheorem{remark}[theorem]{Remark}
\newtheorem{lem}[theorem]{Lemma}
\newtheorem{definition}[theorem]{Definition}
\newtheorem{fact}[theorem]{Fact}

\newenvironment{proof}[1][Proof]{\noindent\textbf{#1.} }{\ \rule{0.5em}{0.5em}}

\def\cred{\color{red}}
\def\blk{\color{black}}

\newcommand{\id}{\mathrm{id}}

\def\io{\rm io}

\makeatother

\begin{document}

\title{Constructive nonlocal games with very small classical values}

\author{
M. Rosicka${}^{1,2}$,
S. Szarek${}^{3,4}$,  A. Rutkowski${}^1$, P. Gnaci\'nski${}^1$, M. Horodecki${}^{1,5}$}

\affiliation{
 ${}^1$Institute of Theoretical Physics and Astrophysics \\ Faculty of Mathematics, Physics and Informatics, University of Gda\'{n}sk, 80-952 Gda\'{n}sk, \\  National Quantum Information Centre in Gda\'{n}sk, 81-824 Sopot, Poland}
 \affiliation{
 ${}^2$Institute of  Informatics, 
 University of Gda\'{n}sk, 
 80-952 Gda\'{n}sk
 }
\affiliation{ ${}^3$Case Western Reserve University, Department of Mathematics, Applied Mathematics and Statistics, 10900 Euclid Avenue, Cleveland, Ohio 44106,  USA }
\affiliation{
${}^4$ Sorbonne Universit\'e, Institut de Math\'ematiques de Jussieu-PRG, 4 place Jussieu, 75005 Paris, France}

\affiliation{
${}^5$ International Centre for Theory of Quantum Information, University of Gda\'nsk, 80-308 Gda\'nsk, Poland
}


\begin{abstract}
There are few explicit examples of two-player nonlocal games with a large gap between classical and quantum value. One of the reasons is that estimating the classical value is usually a hard computational task. 
This paper is devoted to analyzing classical values of the so-called linear games (generalization of XOR games to a larger number of outputs).
We employ nontrivial results from graph theory and combine them with number theoretic results used earlier in the context of harmonic analysis to obtain a novel tool -- {\it the girth method} --  allowing to provide explicit examples of linear games with prescribed low classical value. In particular, we provide games with 
minimal possible classical value. We then speculate on the potential unbounded violation, by comparing the obtained classical values with a known upper bound for the quantum value. If this bound can be even asymptotically saturated, our games would 
have the best ratio of quantum to classical value as a function of the product of the number of inputs and outputs when compared to other explicit (i.e. non-random) constructions. 
\end{abstract}



\maketitle


\section{Introduction}
Nonlocal games can be regarded as a particular type  of Bell-type inequalities and it is well known that Bell-type inequalities  \cite{Scarani2009} are at the heart of quantum information science. While being initially only of philosophical interest, they now form the basis for device independent quantum cryptography \cite{Ekert2014,Scarani-DI}  as well as for quantum advantage \cite{Koenig}.  

By a nonlocal game we mean a scenario in which two cooperating players are asked to assign values $a$ and $b$ to certain randomly chosen variables $x\in X$ and $y\in Y$, respectively. Neither of the players knows which variable was chosen for the other one and they are not allowed to communicate with each other. For each pair $x, y$ of variables which may be chosen, we want the~values $a$ and $b$ to satisfy certain constraints. If $a$ and $b$ satisfy the constraints, the players win, otherwise they lose. 

We can distinguish two types of strategies to win a nonlocal game, namely: the quantum strategy, where the players have access to a shared entangled state and can perform local measurements of observables corresponding to the questions, and the  classical one, where only classical shared randomness is permitted. The maximum probability of winning a given game with a classical strategy is called the \textit{classical value} or \textit{classical winning probability} $p_{Cl}$ of the game. Similarly, the \textit{quantum value} (or \textit{quantum winning probability}, $p_Q$) of the game is the highest probability of winning the game using a quantum strategy. 
If $p_{Cl}< p_Q$, we say that the game exhibits a {\em quantum violation}.

It is known that in order to get an arbitrarily large ratio between quantum and classical values of winning probability of a bipartite nonlocal game  (an {\em unbounded quantum violation}), one has to  grow the number of inputs as well as outputs \cite{Tsirelson1993,Acin2006}.  Moreover, it is well known that calculating -- or even estimating --  classical values of nonlocal  games  is a very hard computational task \cite{Kempe2008}. 

Note that in the literature it is more common 
to compare {\em biases} rather than values, but in our setting the two quantities nearly coincide, see Section \ref{subsec:biases}. One should also notice, that, as posited in \cite{single-shot-Araujo}, it is the difference of values 
that is arguably operationally most relevant. In this paper we shall nevertheless use ratio 
as a measure of quantum violation.

At the moment, there are very few explicit examples of games that exhibit unbounded quantum violation (for example, Khot-Vishnoi game  \cite{Khot2005, Kempe2008a, Regev} or Hidden Matching game \cite{Buhrman}). 
Others, such as \cite{Junge2010}, are random constructions. 
Such random constructions -- even if they, as is frequently the case, lead to stronger results -- are not fully satisfactory: for actual implementation one needs explicit objects, and even if a randomly obtained game has the desired properties, certifying that fact may be computationally infeasible.

In this paper we make a step towards such new explicit examples by analyzing classical values of so called linear games \cred (\cite{Zukowski1997}, Section III, or \cite{Murta2016, RRGHS}) \blk. In particular, we describe methods of obtaining constructive examples of games with extremely low classical values. This paves the way towards a large ratio between the classical and quantum values of a game, namely: we provide explicit and transparent methods of constructing games with low classical values with a relatively small number of outputs. Additionally, we have shown that these games have an unbounded ratio  between the classical value and the upper bound on the quantum value found in \cite{RAM}.
By some measures this ratio scales better than those exhibited by any explicit (i.e., non-random) games known so far.  

We also address the problem of calculating classical values of games and provide exact solutions for certain classes of games. 

To achieve the above results we introduce a novel tool  -- {\it the girth method} -- resulting from combining mathematical results from two distant domains of mathematics, namely  (i) {\it   Number theory related to harmonic analysis:} We employ construction of set of natural numbers, introduced by Rudin \cite{Rudin} as a tool for harmonic analysis, which have the sum property, i.e., for any pair of  subsets with number of elements smaller than some fixed $s$ and equal, their sums are never the same;
 (ii) {\it Graph theory:} We use the result saying (quantitatively) that a graph with 
  high girth \blk
 cannot have too many edges - the first and, up to a constant, optimal result  was conjectured by Erd\H os in 1964 and proved in 1974 in \cite{Bondy-girth}.

 In Section \ref{sec:est} we introduce the main framework used in this paper. We connect the classical value of a game to the rigidity of a matrix which represents the game. We provide ways of finding classical values in general, as well as easier methods for a few specific classes of games.
In Section \ref{sec:girth-method} we introduce graphs $H$ and $H_{opt}$ associated with games and study the classical value of games in terms of the girths of the corresponding graphs.
In Section \ref{sec:construction} we provide two different explicit methods for constructing low-dimensional non-local games  with low classical values, as well as attempt to bound the number of inputs necessary to minimize the classical value of a game with a given number of inputs.
In Section \ref{sec:graphs} we combine our present approach, in which games are represented as matrices, with a framework for studying non-local games based on labeled graphs of \cite{RS} and \cite{RRGHS}.

\section{Linear games and preliminaries on their classical values}
\label{sec:est}

In this section we apply the matrix framework to study properties of nonlocal games where a particular nonlocal game is described by means of a matrix. In particular, we ask how the properties of the matrix affect the classical value of the corresponding game. We provide a method of finding the classical value of any matrix in terms of minors of size 2, as well as examples of more efficient methods for specific types of matrices.
Let us start by introducing the basic concepts and notation used in this article.
 
\subsection{Nonlocal  games: }
Let $X,Y,A,B$ be nonempty finite sets. In the spirit of \cite{Cleve2004}, a nonlocal
game  is determined by  $\pi$, a probability distribution
on the Cartesian product $X\times Y$, and $\text{Pred}:X\times Y\times A\times B\rightarrow\left\{ 0,1\right\}$, the
so called predicate function. 
The rules of the game are the following: the referee
will randomly choose questions $\left(x,y\right)\in X\times Y$ according
to probability distribution $\pi$ and send them to two players, 
Alice and Bob. In what follows, the questions will be called {\it inputs}.
 The players will answer the questions with  {\it outputs} $\left(a,b\right)\in A\times B$
without communication. They win the game iff $\text{Pred}\left(a,b|x,y\right)=1$.

The classical value of a game is the maximal winning probability 
while using only classical strategies. We then have

\begin{equation} \label{classical_value}
p_{Cl}=\max_{a,b}\sum_{x,y}\pi\left(x,y\right)\text{Pred}\left(a\left(x\right),b\left(y\right)|x,y\right),
\end{equation}

where the maximum is taken over all functions $a:X\rightarrow A$
and $b: Y\rightarrow B.$  (At the face value, the formula \eqref{classical_value} 
accounts only for {\em deterministic}
strategies,  but in the setting where shared randomness is allowed the 
values $p_{Cl}$  are the same.) \blk
A quantum strategy for the players involves using a quantum state $\rho$ 
(living on the Hilbert space $H_{A}\otimes H_{B}$) 
shared by Alice and Bob, and a quantum measurement for Alice
(respectively for Bob) for each $x\in X$ $\left(y\in Y\right)$.
For every question $\left(x,y\right),$ the probability of the output
$\left(a,b\right)$ is given by:

\begin{equation}
P\left(a,b|x,y\right)=\text{Tr}\left((E_{x}^{a}\otimes E_{y}^{b})\rho\right),
\end{equation}
where $\big(E_{x}^{a}\big)_{a\in A}$ and $\big(E_{y}^{b}\big)_{b\in B}$ are POVMs living on 
$H_{A}$ and $H_{B}$ respectively. 

The quantum value of the game is given by

\begin{equation}
p_{Q}=\sup\sum_{x,y}\sum_{a,b}\pi\left(x,y\right)\text{Pred}\left(a,b|x,y\right)\text{Tr}\left((E_{x}^{a}\otimes E_{y}^{b})\rho\right),
\end{equation}
where the $\sup$ is taken over all possible quantum strategies. 

\begin{example}
	The CHSH game (named after Clauser, Horn, Shimony and Holt) is the
	nonlocal game for which we have: $X=Y=A=B=\left\{ 0,1\right\} $, 
	the probability distribution $\pi$ is given by 
	
	\begin{equation}
	\pi\left(0,0\right)=\pi\left(0,1\right)=\pi\left(1,0\right)=\pi\left(1,1\right)=\frac{1}{4},
	\end{equation}
	and the predicate function is 
	
	\begin{equation}
	{\rm Pred}\left(a,b|x,y\right)=\begin{cases}
	1 & \text{if}\quad a+ b=x y \mod 2\\
	0 & \text{if}\quad a+ b\neq x y \mod 2
	\end{cases}, 
	\end{equation}
	 The classical value of the CHSH game is $p_{Cl}=\frac{3}{4}$, while 
	 its quantum value is $p_{Q}= \cos^2(\pi/8) \eqsim 0.85$.
\end{example}

In \cite{RRGHS} a labeled graph framework is used to represent a broad class of nonlocal games known as unique games. The inputs correspond to vertices of a graph $G$ and a labeling $K:E(G)\mapsto S_n$, assigning permutations of a set of $n$ elements to the edges, defines the correlations between desired outputs.

In the case of linear games on bipartite graphs this approach translates naturally into a matrix-based representation. This connection is described in more detail in Section \ref{sec:graphs}. In particular, a contradiction in a game matrix corresponds directly to a contradiction in a labeled graph.

\subsection{Linear games and their matrix representation:}
We have already provided the definition of a nonlocal game as well as its \textit{classical value} $p_{Cl}$ and its \textit{quantum  value} $p_{Q}$.  In this paper we focus on nonlocal games called \textit{linear games} \cite{Zukowski1997, Murta2016, RRGHS}, for which the predicate $\text{Pred}$ is given by
\begin{equation}
	\label{eq:V}
 \text{Pred}\left(a,b|x,y\right)=\begin{cases}
 1 & \text{if}\quad b=a + k_{xy} \mod d \\
 0 & \text{else}
 \end{cases},
\end{equation}
where $k_{xy}\equiv k(x,y)\in\{1,\ldots,d\}$, and $d$ is the number of outputs (equal for both parties).  We shall also consider only 
uniform probability distribution over inputs 
\begin{equation} 
\pi(x,y)=\frac{1}{n_A n_B},
\end{equation}
where $n_A=|X|$ and
$n_B=|Y|$  are the numbers of inputs for Alice and Bob respectively. 
(In this paper, most of the time  we will have $n_A=n_B=n$.  To avoid trivialities, we will {\em always} assume that $n_A, n_B \geq 2$).
This type of game can be represented as an $n_A\times n_B$ matrix $M$.  
The elements of the matrix $M$ will be complex roots of $1$ of order $d$: 
\begin{equation}
M=\left(\begin{array}{cccccc}
m_{11} & m_{12} & . & . & . & m_{1 n_B}\\
m_{21} & m_{22} &  &   &   & m_{2n_B}\\
. &     & . &   &   & .\\
. &     &	  & . &   & .\\
. &     &   &   & . & .\\
m_{n_A 1} &  .  & . & . & . & m_{n_A n_B}
\end{array}\right),
\end{equation}
where 
\begin{align}
\label{eq:def-m-omega}
m_{xy}= \omega^{k_{xy}},\quad
 \text{with} \quad \omega=e^{i 2 \pi/d}, 
 \end{align}
 where $k_{xy}$'s are the same as in \eqref{eq:V}.
\cred 
We see that the matrix entries are in one-to-one correspondence with 
the parameters $k_{xy}$ defining the game, so that the matrix determines the game. 
We shall later see how the properties of 
a game can be 
extracted from its matrix (e.g. in Fact \ref{contradiciotn_numb}).
\blk  
  
The simplest linear game  is the CHSH game, and it has the following matrix representation:
\begin{equation}
	M_{CHSH}=\left(\begin{array}{cc}
	1 & 1\\
	1 & -1
	\end{array}\right)
\end{equation}	

We shall later use the notions of game and matrix interchangeably.

\subsection{Equivalence of games (matrices).}
Here we introduce important definition of equivalence of games in matrix representation. 
\begin{definition}
\label{def:equiv}
We consider two matrices to be {\it equivalent} if one can be obtained from the other by a sequence of the following operations: 
\begin{enumerate}
	\item[(i)] multiplying a row or column by a root of unity,
	\item[(ii)] swapping two rows or two columns.
	\item[(iii)] transposition
\end{enumerate} 
\end{definition}

The transformations (ii) and (iii)   can be interpreted as relabeling the inputs or outputs of the game. 
The transformation (i) is 
relabeling outputs of each observable separately. 
 Thus the above operations do not change  classical or quantum values of the game. In other words, equivalent games have the same classical (quantum) value.

\begin{remark} \label{order} 
Note that given two equivalent matrices, 
a transformation between them can be always achieved by first applying operations of type (i), then operations  of type (ii), and then (iii), if applicable 
(or any other order of (i), (ii) and (iii)). 
\end{remark}
\begin{definition}
\label{contradiction}
By a contradiction for a given game and a strategy $a(x)$ and $b(y)$ we mean a pair $(x,y)$ such that $b(y) \not = a(x) + k_{xy} \mod d.$
\end{definition}
Let us notice that if we use the above operations we can bring the matrix of any game 
to a ``standard form'' via the following proposition. 
\begin{proposition}
	\label{prop1}
	Every matrix is equivalent to a matrix in which all elements in the first row and column are equal to $1$.
\end{proposition}
This can be easily shown for any matrix via a series of transformations, in which rows and columns are multiplied by appropriate factors (see Appendix).

It follows from Proposition \ref{prop1} that the contradiction number of a game 
corresponding to a $n_A\times n_B$ matrix 
(i.e. with $n_A$ Alice's inputs and $n_B$ Bob inputs)  cannot be greater than $(n_A-1)(n_B-1)$  so that we have
 \blk

\begin{fact}
	\label{fact:min-clas-value}
	For any linear game with $n_A\times n_B$ inputs the classical value satisfies 
	\begin{eqnarray}
	p_{Cl}\geq \frac{n_A+n_B-1}{n_A n_B}.
	\end{eqnarray} 
	Equivalently, the contradiction number satisfies
	\begin{eqnarray} \label{beta_max}
\beta_C \leq (n_A -1)(n_B-1).
	\end{eqnarray}
\end{fact}
An interesting question is under what conditions the inequality in \eqref{beta_max} is saturated. This is studied more in-depth in Section \ref{sec:girth-method}, cf. in particular Lemma \ref{lem:max-contr-distinct}.

\subsection{Classical value of linear games and contradiction number.}

Let us start with several basic definitions and a few elementary observations.

\begin{definition}
\label{contradiciotn_numb}
Contradiction number is  the minimum number of contradictions over all classical strategies and it is denoted by: $\beta_C(M)$.

\end{definition}
\begin{fact} {\cred {\rm (\cite{RRGHS})}} 
\label{classical_val}
The classical value of a game defined by an $n_A\times n_B$ matrix $M$ can be expressed as
follows \blk
\begin{equation}
p_{Cl}(M)=1 - \frac{\beta_C(M)}{n_An_B} .
\end{equation}
\end{fact}

\begin{fact}
\begin{enumerate}

\item Games whose matrices are equivalent have the same contradiction number. 
\item 
A game with no contradictions corresponds to a matrix of rank $1$, which is equivalent to one in which all elements are equal to $1$.
\item
The contradiction number $\beta_C(M)$ of a game matrix $M$ is no smaller than the number of entries which need to be changed in order to obtain a rank 1 matrix.
\end{enumerate}
\blk
\end{fact}
\begin{proof}
 Item 1 can be 
    can be shown easily using the graph framework used in \cite{RRGHS} (see Section \ref{sec:graphs}).
    For completness we present here also a direct argument.
For a given classical strategy given by numbers $\{a_i\}_{i=1}^n$ and 
$\{b_i\}_{i=1}^n$ the number of contradictions is given by the number of sites $(i,j)$ such that 
$a_i\not= b_i+k_{ij} \mod d$.  Alternatively, 
if we define $\xi_i=\omega^{a_i}$ and $\eta_j=\omega^{b_j}$, 
with $\omega$ being $d$-th root of unity, the condition for contradiction at site $(i,j)$ is 
\begin{align}
    \xi_i \not  =M_{ij} \eta_j
\end{align}
Now, if we swap rows or columns of the matrix, 
then the classical strategies $\xi'_i$ and $\eta'_i$
each with swapped corresponding elements will give the same number of contradictions. 
Transposition of $M$, in turn corresponds to 
swapping Alice strategy $\xi$ with Bob's one $\eta$. 
Finally, if we multiply  a row $i_0$ with root of unity $\omega$, then 
the strategy $\xi$ with  $\xi_{i_0}$ multiplied  with $\omega^{-1}$ gives contradictions for the same sites. 

Let us now prove item 2.
To this end consider matrix $M$ without contradictions. 
This means that there are $\xi$ and $\eta$ such that 
for all sites $\xi_i=M_{ij}\eta_j$.  
Then equivalent matrix given by 
\begin{align}
\label{eq:equiv-opt-eta-xi}
    \tilde M_{ij}=
    \xi_i^{-1} M_{ij} \eta_j
\end{align}
has all entries equal to $1$. 

We now proceed with the proof of item 3. 
Suppose matrix has minimal number of contradictions equal to $l$. Then by changing $l$ entries (those which give rise to contradictions) we obtain matrix without contradictions, hence of rank one by Item 2. 
This means that the number of entries that have to be changed cannot be larger than $l$. 
\end{proof}

The proof of the above facts 
lead us to the following result, that will be later used by us 
to develop girth method.

\begin{proposition} {\cred {\rm (\cite{RRGHS})}} 
\label{prop:beta-non-ones}
The contradiction number $\beta_C(M)$ is equal     
to the minimal number of 
non-one entries 
over all equivalent matrices
\blk. 
\end{proposition}
\begin{remark}
    Since flipping rows or columns, or transposition do not change number of $1$'s, 
    $\beta_C(M)$ is actually equal to minimal number of non-one entries over matrices obtained just by multiplying rows and columns with a root of unity.
\end{remark}
\begin{proof}
First of all, clearly $\beta_C(M)$ is no larger than 
the minimal number of non-$1$'s 
over equivalent matrices. 
Indeed, equivalence does not change $\beta_C$, so consider 
any equivalent matrix $M'$, and let us check that $\beta_C(M')$
is no larger than number of non-one entries in $M'$. It is indeed so, because we can consider a trivial strategy, where all $a_i=0$ and $b_j=0$ (equivalently $\eta_i=1,\xi_j=1$), then entries of $M'$ equal to $1$ do not create contradictions. 

Now we pass to showing that $\beta_C$ is no smaller than 
the minimal number of non-$1$'s 
over equivalent matrices. 
To this end, consider optimal strategy given by $\eta$ and $\xi$,
i.e. now $\beta_C(M)$ is equal to the number of contradictions for this strategy.
We then consider equivalent matrix given by 
\eqref{eq:equiv-opt-eta-xi}.
This matrix has $1$'s for each entry of original matrix, which didn't lead to contradiction. 
Thus $\beta_C(M)$ is equal to the number non-$1$'s in the new matrix. 
Since this was some particular way of obtaining non-ones by 
equivalence transformations, 
it proves that $\beta_C$ is upper bound for minimal number of non-ones over equivalent matrices. 
\end{proof}
\blk

At the end of this section, let us make connection between 
contradiction number and notion of {\it rigidity} of a matrix. Namely, 
 Item 3 implies that 
 the contradiction number $\beta_C(M)$ of a game matrix $M$ is no smaller than \blk 
the rigidity ${\rm Rig}(M,1)$ which is defined as follows:
\begin{definition}
\label{rigidity}
	For a matrix $M$, let the weight ${\rm wt}(M)$ be the number of nonzero elements within the matrix. The rigidity ${\rm Rig}(M,k)$ is defined by	
	\begin{center}
		\begin{equation}
		{\rm Rig}(M,k)=\min_{M'}\{{\rm wt}(M-M')\left| {{rank}}(M')\leq k\right.\}
		\end{equation}
	\end{center}
	\noindent where $M'$ ranges over all matrices of the same size as $M$.
\end{definition}

In other words, the rigidity ${\rm Rig}(M,k)$ of a matrix is the number of elements which need to be changed in order to obtain a matrix of rank at most $k$. 
The concept of rigidity was introduced in \cite{Valiant77}. The problem of calculating rigidity has been studied in papers such as  
\cite{Friedman93} \cite{Kashin98} and \cite{Lokam01}, largely in terms of providing lower and upper bounds for certain types of matrices. In our approach we have not used rigidity, but in our opinion it should be explored more in the context of linear games.

\subsection{Contradiction number via minors}
As mentioned before, the classical value of a game is closely tied to the contradiction number of the matrix and $\beta_C(M)={\rm Rig}(M,1)$. We will now show how this can be used to calculate the classical value of a game defined by a matrix.

\cred
The rank of a matrix is the largest order (size) of a non-zero minor that can be found within the matrix. More formally, the rank is the maximum integer \( r \) such that there exists at least one \( r \times r \) submatrix whose determinant (minor) is non-zero.  If all minors of size $2$ are equal  to zero, there can be no nonzero minor of size $3$ or more and the rank of the matrix is $1$. \blk Thus, we have the following
\begin{remark}
A matrix contains no contradiction iff all minors of size 2 are equal to $0.$  
\end{remark}
It is also immediate to see that 
\begin{remark} \label{minor2minor}
The equivalence operations of Def. \ref{def:equiv} transform nonzero minors into non-zero minors.  Consequently, such operations also transform submatrices of $M$ into submatrices of $M'$ of the same rank.
\end{remark}
A minor of size $2$  of $M=\big(m_{ij}\big)$ can be expressed as
\begin{equation}
m_{ij}m_{st}-m_{it}m_{sj}.
\end{equation}
This means that the minor equals $0$ iff
\begin{equation}
\label{minors}
m_{ij}m_{st}=m_{it}m_{sj},
\end{equation}
%
Since the contradiction number of a matrix is the number of elements we must change in order to obtain a matrix without nonzero minors of size $2$, it follows that the contradiction number of a matrix is the number of elements which need to be changed in order to obtain a matrix which satisfies equation (\ref{minors}) for all $i,j,s$ and $t$ 

or, equivalently,
%
\begin{equation}
\label{minorsplus}
k_{ij}+k_{st}=k_{it}+k_{sj}  \mod d,
\end{equation}

\noindent where as in \eqref{eq:def-m-omega} $m_{ij}=\omega^{k_{ij}}, m_{st}=\omega^{k_{st}},m_{it}=\omega^{k_{it}},m_{sj}=\omega^{k_{sj}}$ and $k_{ij},k_{st},k_{it},k_{sj}\in \textbf{Z}_d$.

Since the number of equations is polynomial in $n$, this gives us a good method for 
for checking whether $\beta_C(M)=0$. Variations of this trick help when the matrix has a special structure. 
Unfortunately, the number of solutions to these equations increases exponentially in $n$  hence it is hard to deal with analytically \blk.

In some specific cases it is possible to determine the classical value much more easily, using the following results.

\subsection{Game with only diagonal nontrivial elements in matrix representation}
\label{sec:diagonal}

Notice first that every $n\times n$ matrix with at most one element not equal to 1 in each row and in each column is equivalent in the sense of Definition 
\ref{def:equiv} to

\begin{equation} \label{diagonal}
M=\left(\begin{array}{cccccc}
m_{00} & 1 & . & . & . & 1\\
1 & m_{11} & 1 &   &   & 1\\
. &     & . &   &   & .\\
. &     &	  & . &   & .\\
. &     &   &   & . & .\\
1 &  .  & . & . & 1 & m_{n-1 n-1}
\end{array}\right),
\end{equation}
where $m_{ii}= 1$ for $0\leq i\leq l < n$ and $m_{ii}\neq  1$ for $l< i\leq n $ (for some $l\in \{0,1,2,\ldots,n\})$. Recall that this equivalence preserves both classical and quantum values. Thus, showing the following for the above matrix $M$ also proves it for all matrices equivalent to $M$.

\begin{proposition}
\label{diag}
Let $n \geq 4$ and let $M$ be an $n\times n$ matrix, for which the only elements different from $1$ are on the diagonal. The contradiction number of $M$ is equal to the number of elements different from $1$. 
\end{proposition}

For $n=3$ the above assertion does not hold, see the example in Proposition \ref{prop:diag3by3}.

\medskip
\begin{proof}  We can assume that the matrix $M$ is of the form \eqref{diagonal}. 

We shall base on the Proposition \ref{prop:beta-non-ones}
which says, that $\beta_C$ is the minimal number of entries not equal to $1$ over all equivalent matrices. 

Therefore, first of all 
$\beta_C(M) \leq n-l$, the number of entries of $M$ different from $1$.  
\blk

For the opposite inequality, suppose there is a matrix $M'$ equivalent to $M$ with larger number of $1$'s than $M$.   {\it A fortiori}, such $M'$ must have at least $l+1$  ``free'' rows and columns (i.e.  rows and columns containing no elements different from $1$). 

Moreover, since 
swapping rows/columns and a transposition do not change the number of $1$s, 
we can assume that  $M$ is transformed into $M'$ only by multiplying rows and columns by roots of unity (cf. Remark \ref{order}). 

We now consider two cases. 

\smallskip\noindent Case $1^\circ$:  \ $l \geq 1$. We have then at least $l+1 \geq 2$ ``free'' columns of $M'$ and exactly $l$ ``free'' columns of $M$. 
Consequently, there is a pair of columns, say the $j$th and the $k$th, such that the $n\times 2$ submatrix of $M'$ consisting of the $j$th and the $k$th column 
has all entries equal to $1$ while the corresponding  $n\times 2$ submatrix of $M$ has one or two entries (namely, $m_{jj}$ or $m_{kk}$, or both) that are different from $1$. This is impossible since every $2\times 2$ minor of the submatrix of $M'$ is equal to zero, while there must be a nonzero  $2\times 2$ minor of the submatrix of $M$, contradicting  the observation from Remark \ref{minor2minor}. The nonzero minor uses one of the $j$th row if $m_{jj}\neq 1$ (and the $k$th row otherwise) and any row other than the $j$th and the $j$th. 
(This part of the argument requires only $n\geq 3$.) Equivalently, we could argue that the above submatrix of $M'$ is of rank $1$, while the corresponding submatrix of $M$ is of rank $2$. 

\smallskip\noindent Case $2^\circ$   \ $l = 0$. We have then a pair of columns, say the $j$th and the $k$th, such that the $n\times 2$ submatrix of $M'$ formed by them contains at most one entry different from $1$. Consequently, there is a further $(n-1) \times 2$ submatrix that consists exclusively of $1$s and so all its $2\times 2$ minors are equal to zero. At the same time, since $l=0$ and so every column of $M$ contains exactly one element different from $1$, the corresponding $(n-1) \times 2$ submatrix of $M$ contains at least one (and possibly two, in different rows) element different from $1$, which is again impossible by the same argument as in Case $1^\circ$ as long as $n-1\geq 3$, i.e. $n\geq 4$, concluding the argument.
%

%
\end{proof}

\subsection{Game with nontrivial elements only in one row (or column) in matrix representation}
\label{sec:row}

Let us now consider an $n_A\times n_B$ matrix in which all elements different from $1$ are in the same row (or in the same column).

\begin{proposition} \label{prop:one_row}
If $M$ is an $n_A\times n_B$  matrix 
\vspace{1ex}
\begin{equation}\label{last_row}
M = \left(\begin{array}{ccccccc}
1 & . & . &  .  &  1\\
. & .  & . &  .   & .\\
. & .  & . &   .  & .\\
1 & .  & . &  .   & 1\\
m_1 &.& . &   . & m_{n_B}
\end{array}\right) 
\end{equation}
\vspace{1ex}

\noindent

such that the number of the most common elements in the last row is $k$,

then $\beta_C(M)=n_B-k.$
\end{proposition}

 Note that any matrix in which all elements not equal to 1 are in the same row or in the same column is equivalent to a matrix of form \eqref{last_row}.

\smallskip \begin{proof}  First, if the last row of $M$ contains $k$ elements equal to $z$, then multiplying the last row by $z^{-1}$ we create $k$ entries equal to $1$, which shows that $\beta_C(M)\leq n_B-k.$ 

For the opposite inequality, suppose $\beta_C(M)\leq n_B-k-1$. As in the proof of Proposition \ref{diag}, this means that there is a matrix $M'$ obtained from $M$ by multiplying rows and columns by roots of unity, which contains at most $n_B-k-1$ elements different from $1$. {\it A fortiori}, there are at least $k+1$ ``free'' columns of $M'$ that consist exclusively of $1$s. Consequently, there are two such columns, say the $j$th and the $l$th, for which $m_j\neq m_l$. (Recall that no more than $k$ of the $m_j$s can take the same value.) But this is impossible: every $2\times 2$ minor of the $n_A\times 2$ matrix consisting of the $j$th and the $l$th column of $M'$ is equal to zero, while obviously there are minors of the corresponding $n_A\times 2$ submatrix of $M$ that are nonzero (namely, minors involving the last row), contradicting  the observation from Remark \ref{minor2minor}. 

\end{proof}

\section{The girth method for bounding classical values}

\label{sec:girth-method}
In this section we analyze the connection between the classical value of a game and the cycles within certain graphs $H(M)$ and $H_{opt}(M)$ derived from the matrix which defines the game. It turns out the length of cycles permitted in $H_{opt}$ gives us control over the number of contradictions in $M$ and thus the classical value of the game. 
We call it {\it the girth method}, as the girth of a graph is the length of its shortest cycle, 
and our method will be precisely to construct games with controlled girth. (We refer to \cite{BondyMurty} for basic terminology and results of graph theory.)

\begin{definition}[Graph of a game]
For any $n_A\times n_B$ game matrix $M=\{m_{ij}\}$ we define $H(M)$ as a bipartite graph with the vertex set $V(H)=\{1_A,2_A...,n_A,1_B,2_B,...,n_B\}$. Two vertices  $i_A,j_B$ in $H$ are adjacent in $H$ iff $m_{i_A j_B}=1$ in $M$. 
\end{definition}

The graph $H(M)$ can also be constructed from the corresponding labeled graph $(G,K)$. The set of vertices is simply $V(G)$ and an edge $e\in E(G)$ belong to $E(H)$ if and only if it is assigned the identity by the labeling $K$.

\begin{definition}[Optimal graph of the game]
Let $M$ be any game matrix and let $M_{opt}$ be a matrix equivalent to $M$ with the maximum number of elements equal to $1$. We will refer to the graph $H_{opt}(M)=H(M_{opt})$ as an optimal game graph of $M$.
\end{definition}
The optimal game graph is not necessarily unique, but the considerations that follow are not affected by how we make the selection. 
\begin{remark}
\label{rem:no-permutations-needed}
Note that permutations of columns and rows, as well as transposition of matrix  do not change the number of $1$'s in a matrix. 
Hence an optimal matrix can be obtained from the original matrix solely by multiplying rows and columns  by roots of $1$. 
\end{remark}

We shall now use  Proposition 
\ref{prop:beta-non-ones} to connect 
the classical value of the game associated with $M$ 
with the properties of the graph $H$. 

\begin{fact}
\label{fact:cl-H}
For any $n\times n$ game matrix $M$
\begin{equation}
\label{eq:cl-edges}
p_{Cl}=\frac{m}{n^{2}}
\end{equation}    
 where $m$ is the number of edges  in $H_{opt}(M)$.   
\end{fact}

\begin{proof}
From Prop. \ref{prop:beta-non-ones} we know that $\beta_C$ is given by minimal number of non-$1$ entries 
over all equivalent matrices. 
Thus, by definition of $H_{opt}(M)$, it
equal to the number of missing edges in the graph, 
i.e. $n^2-m$.
Due to Fact \ref{classical_val}
we have  $p_{Cl}=1-\frac{\beta_C(M)}{n^{2}}$, so inserting $\beta_C=n^2-m$
we obtain the result.
\end{proof}
\blk

 Accordingly, to get games with low classical value we need to construct matrices which correspond to graphs $H_{opt}$ with small number of edges. On the other hand, it is well known that a graph which does not have short cycles cannot have too many edges. Hence our aim will be to construct matrices such that the corresponding graphs do not have short cycles. 
We shall illustrate this idea by means of the following example.

\begin{center}
\begin{figure*}[!htbp]
\begin{picture}(500,250)
	\put(0,-10){\includegraphics[scale=0.50]{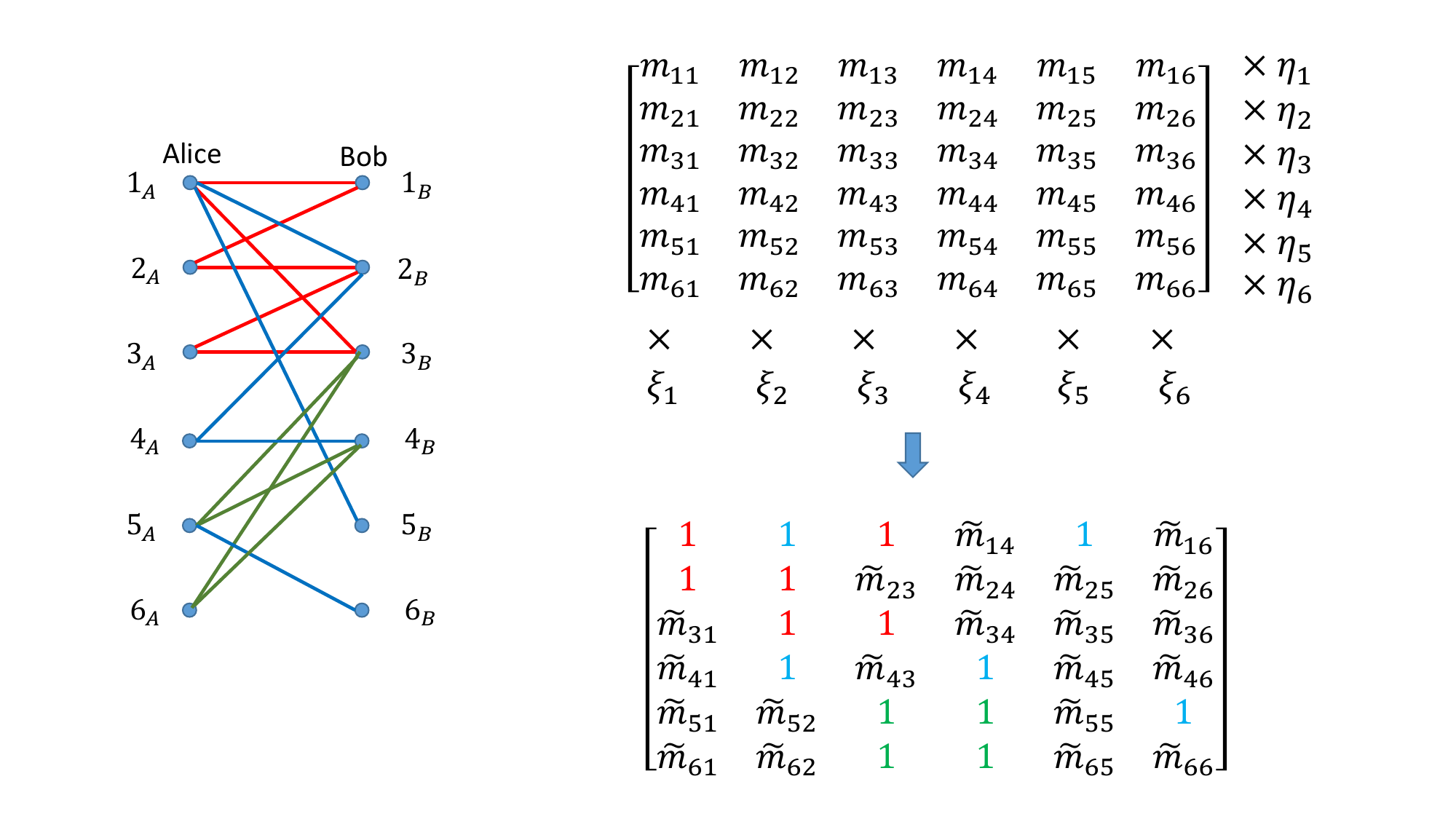}}
\end{picture}
	\caption{Cycles in graph $H$ and the matrix of the game. There are two cycles, marked in green and red.}
	\label{fig:cycle}
\end{figure*}
\end{center}

{\it Example.}
Let us note that every cycle in the graph $H_{opt}$  corresponds to a family of elements in the matrix $M$, to which we will also refer as a cycle. 
For instance in Fig. \ref{fig:cycle} the cycle $1_A \to 1_B \to 2_A \to 2_B \to 3_A \to 3_B \to 1_A$ in the graph 
(the red cycle) is related to the cycle  
$m_{11}\to   m^T_{12}\to m_{22}\to m^T_{23}\to m_{33}\to m^T_{31}$ \blk in the matrix  (here $m_ij^T=m_{ji}$)\blk. 
The green cycle $5_A \to 3_B \to 6_A \to 4_B\to 5_A$
corresponds to $m_{53} \to m^T_{36}\to m_{64} \to m^T_{45}$. \blk 

Since the matrix $M_{opt}$ is equivalent to $M$ and can be obtained by multiplying rows of $M$ by $\eta_i$ and columns by $\xi_j$, 
we have  $m_{ij}\, \eta_i\, \xi_j=1$ for all pairs $(i,j)$ such that $a_ib_j$ is an edge in $H_{opt}$.  Hence  in the red cycle we have  
\begin{equation}
\begin{aligned}
(\eta_1 m_{11}  \xi_1)(\xi_1 m^T_{12} \eta_2)^{-1} 
(\eta_2 m_{22}  \xi_2)(\xi_2 m^T_{23} \eta_3)^{-1} 
\\(\eta_3 m^T_{33}  \xi_3)(\xi_3 m^T_{31} \eta_1)^{-1}= 1  
\end{aligned}
\end{equation}
which gives
\begin{equation}
m_{11} (m^T_{12})^{-1}  m_{22} (m^T_{23})^{-1} m_{33} (m^T_{31})^{-1} =1
\end{equation}
or equivalently 
\begin{equation}
m_{11} (m_{21})^{-1}  m_{22} (m_{32})^{-1} m_{33} (m_{13})^{-1} =1
\end{equation}
which, since $m_{ij}=\omega^{k_{ij}}$, means 
\begin{equation}
k_{11}+k_{22}+k_{33}=k_{21}+k_{32}+k_{13} \mod d
\end{equation}
Thus the sum of three numbers $k_{ij}$ is equal to the sum of three other numbers. 

\medskip \noindent {\it The general case.}
In general, we will define a cycle in a matrix $M$ as a set of matrix elements corresponding to a cycle in the complete bipartite graph $K_{n_A,n_B}.$ 
\begin{definition}[Cycle]
	A cycle $C$ in matrix $M$ is a subset of matrix entries determined by ordered sets of rows $S_A=(i_1,\ldots i_l)$ and columns $S_B=(j_1,\ldots j_l)$, $l\geq2$, as follows
\begin{align}
\{ m_{i_1,j_1}, m_{j_1,i_2}^T, m_{i_2,j_2}, m_{j_2,i_3}^T, m_{i_3,j_3} \ldots m_{i_l,j_l}, m_{j_l, i_1}^T \},
\end{align}

\noindent where $m_{ij}^T=m_{ji}.$
\end{definition}

\begin{definition}[Good cycle] 
\label{def:good-cycle}

A cycle $C$ in matrix $M$ given by subsets 
$S_A=(i_1,\ldots i_l)$ and  $S_B=(j_1,\ldots j_l)$ of rows and columns respectively, is referred to as a  {\rm good cycle} 
if it satisfies 
\begin{eqnarray}
\label{eq:good-cycles-matrix}
&&m_{i_1,j_1}  (m_{j_1,i_2}^T)^{-1} m_{i_2,j_2} (m_{j_2,i_3}^T)^{-1} m_{i_3,j_3} \times  \nonumber \\
&&\times \ldots \times m_{i_l,j_l} ( m_{j_l, i_1}^T)^{-1} =1  
\end{eqnarray}
equivalently
\begin{equation}
\label{eq:ks}
\begin{aligned}
k_{i_1, j_1} + k_{i_2,j_2} + k_{i_3,j_3} +\ldots+ k_{i_l,j_l} = \\ k_{j_1,i_2}+ k_{j_2,i_3} + k_{j_3,i_4} +\ldots+ k_{j_l,i_1} \mod d
\end{aligned}
\end{equation}
where $M=(m_{ij})=(\omega^{k_{ij}})$ with $\omega=e^{2 \pi i/d}$. 
\end{definition}
Note that a good cycle in a matrix $M$ corresponds to a good cycle in the labeled graph framework discussed in Section \ref{sec:graphs}.

In the next proposition we  observe that a cycle in $H_{opt}$ corresponds to a good cycle in $M$. This correspondence is direct if the optimal matrix leading to $H_{opt}$ was obtained without  permutations and transposition (i.e. solely by multiplying rows and columns with roots of unity). Otherwise, 
the two cycles are related by permutations or transposition (the latter acts as reflection).

\begin{proposition}
\label{prop:goodcycle}
Consider a cycle $C$ in matrix $M$ given by subsets 
$S_A=(i_1,\ldots i_l)$ and  $S_B=(j_1,\ldots j_l)$ of rows and columns respectively. 
If the graph $H_{opt}(M)$ contains  the corresponding cycle then the cycle $C$ in $M$ 
is a good cycle. 

Conversely, if 
matrix $M$ contains a good cycle, then there exists equivalent matrix $M'$ such that $H(M')$ contains the corresponding cycle.
\end{proposition}
\begin{proof}
Let $M$ be a game matrix and let $C$ be the cycle defined by $S_A=(i_1,\ldots i_l)$ and  $S_B=(j_1,\ldots j_l)$.
As was noted in Remark \ref{rem:no-permutations-needed}, optimal matrix can be obtained 
without permutations of rows/columns or transposition. 
Thus, 
$C$ corresponds to a cycle in so obtained $H_{opt}$ if and only if there exists an optimal matrix $M_{opt}$ in which all elements on the cycle determined by the cycle from $H_{opt}$  are equal to $1$. We can obtain such $M_{opt}$ through multiplying each row of $M$ by some $\eta_i$ and each column by some $\xi_j$, where $i$ and $j$ are the indices of rows and columns, respectively. Thus on the cycle $C$ we have

\begin{align} 
\label{eq:proofgood}
&m_{i_1,j_1}  (m_{j_1,i_2}^T)^{-1} m_{i_2,j_2} (m_{j_2,i_3}^T)^{-1} m_{i_3,j_3}\times\ldots\times  \nonumber\\
&\times m_{i_l,j_l} ( m_{j_l, i_1}^T)^{-1} =
\eta_{i_1}m_{i_1,j_1}\xi_{j_1} \times \nonumber\\
&\times (\eta_{i_2}m_{j_1,i_2}^T\xi_{j_1})^{-1} \eta_{i_2}m_{i_2,j_2}\xi_{j_2} (\eta_{i_3}m_{j_2,i_3}^T\xi_{j_2})^{-1} \eta_{i_3}m_{i_3,j_3}\xi_{j_3}\times \nonumber\\  
&\times \ldots \times \eta_{i_l}m_{i_l,j_l}\xi_{j_l} ( \eta_{i_1}m_{j_l, i_1}^T\xi_{j_l})^{-1} = 1 \times\ldots\times 1 = 1.
\end{align}
which means that the cycle $C$ is a good one. 
\blk

Let us now prove the second assertion.
Suppose that matrix $M$ has  a cycle 
\begin{align}
\{ m_{i_1,j_1}, m_{j_1,i_2}^T, m_{i_2,j_2}, m_{j_2,i_3}^T, m_{i_3,j_3} \ldots m_{i_l,j_l}, m_{j_l, i_1}^T \}
\end{align}
 which is good, i.e. 
if we multiply the above numbers they will give 1, as in  \eqref{def:good-cycle}. 
We shall now consider multiplication of rows and columns with $\xi$'s and $\eta$'s, and will show that 
these numbers can be chosen in such a way, that all the matrix elements of the cycle in the resulting matrix $M'$ are equal to 1. To this end, note that 
these matrix elements look as follows 
\begin{align}
&& m'_{i_1,j_1}= \eta_{i_1}m_{i_1,j_1}\xi_{j_1}  \nonumber\\
&& {m'}^T_{j_1,i_2}=\eta_{i_2}m_{j_1,i_2}^T\xi_{j_1}
\nonumber \\
&& m'_{i_2,j_2}=\eta_{i_2}m_{i_2,j_2}\xi_{j_2}
\nonumber \\
&& {m'}^T_{j_2,i_3}=\eta_{i_3}m_{j_2,i_3}^T\xi_{j_2}
\nonumber\\
&& m'_{i_3,j_3}=\eta_{i_3}m_{i_3,j_3}\xi_{j_3}
\nonumber\\  
&&\cdots 
\nonumber\\
&& m'_{i_l,j_l}=\eta_{i_l}m_{i_l,j_l}\xi_{j_l} 
\nonumber \\
&& {m'}^T_{j_l,i_1}=\eta_{i_1}m_{j_l, i_1}^T\xi_{j_l}
\end{align}
We choose $\eta_{i_1}=1$ and 
$\xi_{j_1}=m_{i_1,j_1}^{-1}$. 
This gives $m'_{i_1,j_1}=1$. 
Next, we choose $\eta_{i_2}$ in such a way that ${m'}^T_{j_1,j_2}=1$. 
We can continue this way to make all $m'$ elements of the cycle equal to one, apart from the last one. Since all indices $i_1,\ldots i_l$  are distinct (and similarly $j_1, \ldots j_l$) the above procedure ensures that $\eta_i$'s and $\xi_j$'s are well defined. Indeed, we have always freedom (except of the last element)  to choose either some $\xi$ or $\eta$ that was not fixed before. 
For the last element, we do not have freedom 
because, $\eta_{i_1}$ and $\xi_{j_l}$ that appear there, have been already fixed. 
But since the cycle was good
(i.e. elements of cycle in $M$ satisfy 
\eqref{eq:good-cycles-matrix}), 
then
according to \eqref{eq:proofgood} the 
elements in corresponding cycle of $M'$
also satisfy  \eqref{eq:good-cycles-matrix}. 
Therefore, since all those matrix elements apart from the last one are equal to $1$, then the last one  must be also equal to $1$. 
Thus for our choice of $\xi$'s and $\eta$'s, the cycle in the matrix $M'$ 
consists of $1$'s. Then by definition  of $H$ the graph $H(M')$ contains the corresponding cycle.
This completes the proof.

The procedure of assigning values to $\eta_i$ and $\xi_j$ can be visualised on the complete bipartite graph. Consider e.g. red cycle in graph from Fig. \ref{fig:cycle}. We make it directed (e.g. first edge let be directed to the right): 
$1_A \to 1_B \to 2_A \to 2_B \to 3_A \to 3_B \to 1_A$.     
The first assignment ($\eta_{1}=1$) is arbitrary. 
the next one,  namely $\xi_1=m_{1,1}^{-1}$ corresponds to head of the next arc $1_A\to 1_B$. The next assignment, i.e. that of $\eta_2$ which assures $m_{21}'=1$, corresponds  to the head of the arc $1_B\to 2_A$.
in this way, for each arc, except from the last one, we can 
change the value of matrix element corresponding to that arc to $1$, by assigning value of the parameter corresponding to the head of the arc. Only, for the last arc, we cannot do this, as the parameter associated $\eta_1$ with its head was already set to $1$ at the beginning. However it is automatically one, as the cycle was good. 
\end{proof}

As already said, the optimal graph is non-unique. However we shall later use its property, which does not depend on particular representative. Namely, if for a game matrix $M$ 
equation \eqref{eq:ks} does not hold for any cycle of length up to, say, $2s$, 
then $H_{opt}(M)$
does not have cycles of length $2s$, irrespectively of the choice of $H_{opt}$. 
Therefore, in such case 
the graph $H_{opt}$
cannot have too many edges, ergo -- due to Fact \ref{fact:cl-H} --
the classical value of the game is small.  

We shall actually mostly use the following relaxation of Proposition \ref{prop:goodcycle}, where we will demand that the equality \eqref{eq:ks} is valid for arbitrary equal subsets of entries (and not only for subsets that give rise to a cycle). 
\begin{corollary}
\label{cor:condition-for-s-cycles}
Let $M=(m_{ij})=(\omega^{k_{ij}})$ with $\omega=e^{2 \pi i/d}$ be a game matrix. Suppose further that for any two disjoint subsets $S_1$, $S_2$ of matrix entries with $|S_1|=|S_2|=s$ we have 
 \begin{eqnarray}
 \label{eq:corS}
 \sum_{m_{ij}\in S_1} k_{ij} \not =   \sum_{{m}_{ij}\in S_2} k_{ij} \mod d .
 \end{eqnarray}
Then  the graph $H_{opt}(M)$  does not have cycles of length $2s$. 
\end{corollary}

\begin{proof}
The proof will be by contradiction. 
    Suppose that there exists a cycle of length $2s$ in the graph. Then 
    by Prop. \ref{prop:goodcycle} 
    the corresponding cycle in the matrix is good, i.e it satisfies \eqref{eq:ks} with $l=s$. 
    We now choose $S_1$ to be the entries labelled by the labels of $k$'s from LHS of \eqref{eq:ks} (i.e. $(i_1,j_1) \ldots (i_s,j_s)$) and $S_2$ to be the entries  labelled by $(j_1,i_2) \ldots (j_{s-1},i_s),(j_s,i_1)$. For those two sets, 
    Eq. \eqref{eq:corS} does not hold, hence we obtain the contradiction. 
\blk 
\end{proof}

Using this corollary we shall bound classical value from above in subsequent subsections. 

\smallskip Finally  we shall provide a result that will later allow to bound the classical value from below. 

\begin{proposition}
\label{prop:connected}
For arbitrary game matrix $M$ there is an equivalent matrix $M'$, such that $H(M')$ contains $H(M)$ 
and is connected. In particular each optimal graph  $H_{opt}(M)$ is connected. 
\end{proposition}

\begin{proof}
Suppose $H=H(M)$ can be decomposed into two nonempty disjoint subgraphs $H_1,H_2$. 
Of course, each of these subgraphs must be bipartite, i.e.,  $V(H_i)= A_i\cup B_i$ with 
$A_1\cup A_2 = A$ and $B_1\cup B_2 = B$, where $A,B$ are the parts of $H$. 
Assume that $i_0\in A_1$ and $j_0\in B_2$. 
The decomposition property means in particular that there is no edge between a vertex in $A_1$ and a vertex in $B_2$, so $(i_0,j_0)$ is not an edge of $H$. 

\cred We now define the sequences of multipliers  $(\eta_i)$ and  $(\xi_j)$ that will be used to obtain the matrix $M'$ from 
$M=\big(m_{ij}\big)$.  Let $\zeta=m_{i_0j_0}$ and set  
\begin{eqnarray}
\eta_i = \left\{ \begin{array}{l c l} 1 & \hbox{ if } & i\in A_1\\
\zeta  & \hbox{ if } & i\in A_2\end{array} \right. , \ \ 
\xi_j = \left\{ \begin{array}{l c l} 1 & \hbox{ if } & j\in B_1\\
\zeta^{-1}  & \hbox{ if } & j\in B_2\end{array} \right.. \nonumber \\
\end{eqnarray}
This construction assures that if $M'=\big(\eta_i m_{ij}\xi_j \big)$, then the $(i_0,j_0)$th entry of $M'$ is $1$, while at the same time leaving unchanged all the entries of $M$ that were equal to $1$. This means that $H(M)\subsetneq H(M')$ (a strict inclusion).

This shows immediately that the optimal graph $H_{opt}$ must be connected. On the other hand, if $H(M)$ is disconnected, then the described operation decreases the number of connected components 
of the graph by one. Repeating it a finite number of times we arrive at a matrix $M'$ 
such that $H(M')\supset H(M)$ and $H(M')$ has only one connected component, i.e. it is connected.
\blk

The construction above assumes tacitly that  sets $A_1$ and $B_2$ are both nonempty; if that is not the case, we proceed similarly for $A_2$ with $B_1$. 
(If one of the $A_i$'s  and  one of the $B_i$'s empty, then the graph has no edges and any choice $i_0\in A$, $j_0\in B$ will work.)  
\end{proof} 
 \blk

\subsection{Games with maximal contradiction number}

Let us now consider in more detail games with maximal number of contradictions,  which due to Fact \ref{fact:min-clas-value} is given by $(n_A-1)(n_B-1)$. 
Examples of such games of 
various level of sophistication will be presented in Section \ref{sec:construction}. 
To begin with, let us start with following definition.
\cred \begin{remark}
\textbf{Tree} is a type of graph that is connected and acyclic, meaning it has no cycles. More specifically, a tree satisfies the following properties:

\begin{itemize}
    \item \textbf{Connected}: There is a path between every pair of vertices.
    \item \textbf{Acyclic}: It contains no cycles, meaning there is no way to start at a vertex, follow edges, and return to the same vertex without retracing any edge.
\end{itemize}

\end{remark}
\blk

Now we make the following simple but enlightening observation.

\begin{fact}
\label{fact:max-contradiction-tree}
An $n_A \times n_B$ game matrix $M$ has the maximal number of contradictions, i.e., $\beta_C= (n_A-1)(n_B-1)$, if and only if 
the optimal graph $H_{opt}(M)$ is a tree.
\end{fact}
\begin{proof}
If $\beta_C= (n_A-1)(n_B-1)$, then, by definition of the optimal graph,  $H_{opt}$  has $n_A+n_B -1 $  edges. By Proposition \ref{prop:connected}, the optimal graph is connected. But arbitrary connected graph of $n$ vertices which has $n-1$ edges must be a tree. The same calculation gives the reverse argument (Proposition \ref{prop:connected} is not even needed). 
\end{proof}
\begin{remark}
An example of such optimal graph for a matrix $M$ with maximal contradictions 
is the following. By  Proposition \ref{prop1}, for any matrix $M$  we can find equivalent matrix  $M'$ to $M$ 
which has $1$'s in  the first row and in the first column, hence it has  at least  $n_A +n_B - 1 $ of  $1$'s. 
If $M$ has maximal possible contradiction number, there does not exist equivalent matrix with  more $1$'s. Thus 
the $H(M')$
is optimal for $M$, and it is a tree. 
\end{remark}

We can now formulate a corollary, which characterizes games possessing the maximal number of contradictions. 
This result can be also proven within the approach of labeled graphs of 
\cite{RS} and \cite{RRGHS} 
(cf. Lemma \ref{lcycles}).

\begin{corollary}
\label{cor:good-cycle}

A game matrix $M$ has the maximal number of contradictions 
given by $(n_A-1)(n_B-1)$
if and only if 
it does not contain any good cycle.
\end{corollary}
\begin{proof}
If $M$ has a good cycle, then
there is equivalent matrix $M'$ such that  $H(M')$ contains a cycle. By Proposition \ref{prop:connected} there is  equivalent matrix $M''$ such that $H(M'')$ contains that cycle, and is moreover connected.  
Thus the 
optimal graph has at least $n_A+n_B$ edges, so that the number of contradictions cannot be the
maximal one. 
Indeed, connected graphs of $n$ vertices which do not contain any cycles are by definition trees, and a tree with $n$ vertices contains exactly $n-1$ edges. 

Suppose now that $M$ does not have any good cycle. Then Proposition \ref{prop:goodcycle} asserts that 
$H_{opt}$ does not have any cycle, hence it is a tree and by Fact \ref{fact:max-contradiction-tree} the matrix has the maximal number of contradictions
\end{proof}

Finally, let us consider matrix $M$ that has $1$'s in the first row  and the first column, and 
suppose that it has the maximal number of contradictions. We then will argue, that all the other matrix elements has to be distinct from each other, and distinct from 1:

\begin{lem}
\label{lem:max-contr-distinct}
\label{lll}
If a  $n_A \times n_B$ game matrix of the form

\begin{equation}
\label{eq:ones}
M = \left(\begin{array}{cccccc}
1      & 1      & ... & 1\\ 
1      & m_{22} & ... & m_{2 n_B}\\
\vdots & \vdots &     & \vdots\\
1      &m_{n_A2}& ... & m_{n_An_B}
\end{array}\right)
\end{equation}
has  $(n_A-1)(n_B-1)$ contradictions then  all elements $m_{ij}$, for $2\leq i\leq n_A$ and $2\leq j\leq n_B$, are distinct and different from $1$.
\end{lem}
\begin{proof}
A direct proof by matrix transformations is provided in the Appendix.  Here we provide an argument which utilizes the tools introduced in present section.

Suppose that there are two elements among $m_{ij}$,  equal to $v$.
First, consider the case when they are in the same row (if the are in the same column the 
reasoning is analogous). We then look at submatrix 
\begin{eqnarray}
\label{eq:22}
\left(
\begin{array}{ll}
1 & 1\\
v & v \\
\end{array}
\right)
\end{eqnarray}
We see that $ 1 \times v^{-1} \times  v \times 1=1$. 
Thus the matrix contains a good cycle, and by Corollary \ref{cor:good-cycle} it cannot have the maximal number of contradictions. The same argument applies when the identical entries are in the same column. 


Now assume that the two equal elements are in different rows and columns, i.e. $m_{i_1j_1}=m_{i_2j_2}=v.$ We can assume without loss of generality that $i_1<i_2$ and $j_1<j_2.$ In this case $M$ contains the following submatrix

\begin{eqnarray}
\left(
\begin{array}{lll}
1 & \textbf{1} & \textbf{1}\\
\textbf{1} & \textbf{v} & u\\
\textbf{1} & w & \textbf{v}\\
\end{array}
\right)
\end{eqnarray}
Since $ 1 \times v^{-1} \times 1 \times 1 \times v \times 1=1$, the bolded elements of the matrix form a good cycle and thus the contradiction number is again not the maximal one. 

To summarize, we have shown that 
a matrix of the form \eqref{eq:ones} with the maximal contradiction number
must have all the entries $m_{ij}$
with $2\leq i \leq n_A$
and $2\leq j\leq n_B$
distinct. 
Finally we can see 
that those entries 
cannot be equal to 1.  Indeed, 
if any such entry if equal to $1$ 
we have trivially a good cycle formed by 
\begin{eqnarray}
\label{eq:22-1s}
\left(
\begin{array}{ll}
1 & 1\\
1 & 1 \\
\end{array}
\right)
\end{eqnarray}
This can also be seen directly from the definition 
of the contradiction number (Def. \ref{contradiciotn_numb}) since 
taking  classical strategy $a(x)=b(y)=1$ 
leads to the number of contradictions that 
is strictly smaller than $(n_A-1)(n_B-1)$.
\end{proof}

\section{Construction of games with low classical values}
\label{sec:construction}
In this section we provide explicit constructions of games with low classical value. 
In particular, we show that an $n\times n$ matrix with maximum number of contradictions can be obtained with $d\leq 2^nn^{3n}$ outputs.

We also provide a method for constructing such matrices, as well as ones with a large, but not maximum number of contradictions and with the number of outputs that is polynomial in $n$.

\cred
In addition, in subsection A, we give a simple construction of matrix with maximal number of contradictions which 
avoids using most of the subtleties of girth method described in Sec. \ref{sec:girth-method}. However the number of outputs has worse scaling  $d\sim 2^{n^2}$. 
\blk

\subsection{Warmup: a simple construction of games with minimal possible classical values}
\label{sub:max}

In this section we construct a family of games with the minimum classical values in the form of $n_A\times n_B$ matrices with $d=2^{(n_A-1)(n_B-1)-1}$ outputs. 
According to Fact \ref{fact:min-clas-value} such games have classical value
\begin{align}
p_{Cl}=\frac{n_A+n_B-1}{n_A n_B}.
\end{align}
We also show that it is possible to achieve the maximum number of contradictions with a smaller number $d$ of outputs and attempt to find a lower bound on the necessary number of outputs.

It follows from Lemma \ref{lll} that the maximum number of contradictions in an $n_A\times n_B$ matrix cannot be achieved with fewer than $(n_A-1)(n_B-1)+1$ outputs. However, this is not a sufficient condition and, in general, a larger number of outputs is necessary.

We will now sketch a ``brute force'' way to construct a matrix with maximal number of contradictions, using
Corollary \ref{cor:good-cycle}.  
It says 
that in order to maximize the number of contradictions one must ensure 
that
the matrix $M$ contains no good cycles. 
Let us note that according to Definition \ref{def:good-cycle} a good cycle in 
matrix $M$   
can be equivalently characterized as follows:

\begin{enumerate}
\item[(i)] $X = X_1 \cup X_2,$ the two subsets are disjoint and $\left|X_1\right| = \left|X_2\right| \geq 2$;
\item[(ii)] Each column contains either no elements from $X$ or exactly one element from $X_1$ and one from $X_2;$
\item[(iii)] Each row contains either no elements from $X$ or exactly one element from $X_1$ and one from $X_2;$
\item[(iv)] $\sum\limits_{m_{ij} \in X_1} k_{ij} \ { =} \sum\limits_{ m_{ij} \in X_2} k_{ij} \mod d$.
\end{enumerate}
(regarding  item (iv), recall that $m_{ij}=\omega^{k_{ij}}$). 
 It follows that if condition (iv)  fails for every pair of sets $X_1,X_2$ verifying conditions (i)-(iii) (i.e., every cycle is not a good one), then the contradiction number is maximal.

\smallskip{\it Explicit construction of the game.}
The above scheme allows us to apply this approach to matrices of the form \eqref{eq:ones}. Our strategy will be the following: we shall construct a matrix $M$ such that the condition (iv) 
is not satisfied for arbitrary sets $X_1,X_2$ satisfying (i).

Note that  the condition (iv) fails if we take $k_{ij}$ from a sequence 
%
\begin{equation}
\label{k_i}
k_l=2^{l}, \ l=0,1,2,3,\ldots
\end{equation}

This follows from the uniqueness of a representation of an integer in base $2$ and the fact that 
the $0$'s in the sums in the condition (iv) corresponding to the $1$'s from the first row and the first column in  \eqref{eq:ones} can be ignored, because by (ii) and (iii), the set $X$ can contain at most $3$ such elements. 
Essentially the same argument works for any $(k_l)$ such that $k_0 \geq 1$ and $k_l\geq2k_{l-1}$ 
for all $l\geq 1$.

The above scheme allows us to construct $n_A\times n_B$ matrices with the maximal number of  contradictions with 
%
\begin{equation}
d=2^{(n_A-1)(n_B-1)+1}
\end{equation}
 outputs. 

 {\it Example.}
Here is game matrix  with $n_A=4,n_B=3$ inputs
containing  6 contradictions based on the above construction. 
\begin{equation}
    M=\left(\begin{array}{cccc}
1 & 1 & 1 & 1\\
1 & k_{1} & k_{2} & k_{3}\\
1 & k_{4} & k_{5} & k_{6}
\end{array}\right)
\end{equation}
The game has   $2^{7}$ outputs. 
The classical value of this game is the minimal over all $4\times3$ games, i.e.    $\beta_C(M)=6$ and  $p_{Cl}(M)=\frac{1}{2}.$


\subsection{Controlling the contradiction number by the girth of $H_{opt}$}
\label{sub:controlled}

Here we provide a different approach to controlling the number of contradictions within a matrix. In order to ensure a large number of contradictions we will control the length of the cycles permitted in the graph $H_{opt}$. If we make sure that our $k_{ij}$'s are chosen in such a way that any two sums of the same lengths less than or equal to $s$ 
are not equal to one another, we will guarantee that there are no cycles of length less than or equal to $2s$ in $H_{opt}$, i.e. the girth of $H_{opt}$ is more than $2s$.  Since a graph with no short cycles can not have too many edges (a more precise result to that effect is stated later in this section),
and the edges of $H_{opt}$ correspond to  $1$'s in an optimal matrix,
we can use this scheme to obtain a bound on the contradiction number of $M$.


We will now show how to find such $k_{ij}$'s.  
Using Rudin's method from \cite{Rudin} we will construct a set $A$ of integers  which satisfies the so-called {\it $s$-sum property}. Namely, the  sum  of any $t$ elements of the set cannot be equal to the sum of any other $t$ elements for  any $t\leq s$. Then we choose the number $d$ of outputs large enough that none of these sums are greater then $d-1$, thus ensuring that the sums modulo $d$ are also distinct. We can then construct a matrix with a large number of contradictions with elements in the form $\omega^{k_{ij}}$ with $k_{ij}$'s taken from the set $A$.

\medskip

{\it Construction of the set A.} 

Let us recall the method (item 4.7 in \cite{Rudin}, page 219-220) that allows us to construct the set $A$ with number of elements equal to a prime number $p$ 
which satisfies  the {\it $s$-sum property}. 
To this end, consider finite sets $A=A(s,p)$,  with $s=2,3,4,\ldots$
where $s<p$. We now define $\lambda(k),$ for
any  positive integer $k$, by 
\begin{equation}
\lambda(k)\equiv k\:\left(\text{mod}\,p\right),\quad0\leq\lambda(k)\leq p-1,
\end{equation}
and we let $A(s,p)$ be the set consisting of the $p$ integers
\begin{equation}
x_{k}=\sum_{i=0}^{s-1}\lambda\left(k^{s-i}\right)\left(sp\right)^{i}\quad\left(k=0,\ldots,p-1\right).
\end{equation}
The sum reflects the representation of $x_k$ 
in the number system whose base is $sp$ with the leading digit
being $\lambda(k)=k$. 
Hence, we see that $0=x_{0}<x_{1}<\ldots<x_{p-1}$ and
\begin{equation}
x_{p-1}\leq\left(p-1\right)\sum_{i=0}^{s-1}\left(sp\right)^{i}<s^{s-1}p^{s}.
\end{equation}
We now  argue that the so constructed 
set $A(s,p)$ satisfies the $s$-sum property.
Namely, suppose that 
\begin{equation}
y=x_{k_{1}}+x_{k_{2}}+\ldots+x_{k_{s}},\label{eq:sumx_k}
\end{equation}
and express $y$ in base $sp$:
\begin{equation}
y=\sum_{i=0}^{s-1}y_{i}\left(sp\right)^{i}\quad\left(0\leq y_{i}<sp\right).
\end{equation}
Since $\lambda(k)<p$ and hence the sum 
$\sum_{j=1}^{s}\lambda\left(k_{j}^{s-i}\right)  < sp$,  we have
\begin{equation}
y_{i} = \sum_{j=1}^{s}\lambda\left(k_{j}^{s-i}\right)\qquad\left(i=0,\ldots,s-1\right),
\end{equation}
so that
\begin{equation}
\sum_{j=1}^{s}k_{j}^{l}= y_{s-l}\mod p\quad\left(l=1,\ldots,s\right).
\end{equation}

Thus the digits $y_{i}$ of $y$ determine the first $s$ power sums
of $k_{1},\ldots,k_{s}$ in the cyclic field of $p$ elements; hence
they determine the elementary symmetric function of $k_{1},\ldots,k_{s}$
in this field, and this in turn implies that $k_{1},\ldots,k_{s}$ are
determined by (\ref{eq:sumx_k}), up to permutations.  This proves that 
the representation for any $y$ in the form (\ref{eq:sumx_k}) is unique, up to permutations of the $x_{k_{i}}.$
The same  argument works for  any lengths $t$  smaller than $s$.

\smallskip 
Having constructed the family of sets $A(s,p)$ 
we are now ready to formulate 
the main result of this section, contained in the following proposition.
\begin{proposition}
\label{prop:controlled}
For any $n\geq 2$ and $s\leq n$ there exists a game associated with an $n\times n$ matrix with $d\leq (2sn^2)^s$ outputs such that $p_{Cl}\leq 2n^{-1+\frac{1}{s}}.$

Moreover, for any $n\geq 2$ there exists a game associated with  an $n\times n$ matrix with $d\leq 2^n n^{3n}$ outputs such that $p_{Cl}=\frac{2n-1}{n^2}$.
\end{proposition}

\begin{proof}
We aim to construct a class of games with controlled number of contradictions.  
Corollary \ref{cor:condition-for-s-cycles} tells us that, to this end, one should  find $n^2$ integers $k_{ij}$ satisfying the ``mod-$d$'' $s$-sum property, i.e., such that their sums of lengths $t\leq s$ are to be distinct modulo $d$. Thus 
 we will need to have 
a set $A$ which (i) has at least $n^2$ elements and (ii) satisfies the above mentioned ``mod-$d$'' $s$-sum property.  
We will achieve property (i) by choosing a suitable prime $p$ in Rudin's construction described above, while (ii)  will assured by choosing $d$ large enough. 
Namely,  let us begin with a set $A=A(s,p)$, where $s\geq 2$ and  $p>s$ is a prime. 
By the argument above, the set $A$ consists of $p$ numbers that satisfy the $s$-sum property.
The numbers are not larger than $s^{s-1}p^s-1$. 
Now, to have ``mod-$d$'' $s$-sum property, we will choose $d$ to be larger than any of the sums.  To this end, we take $d=s \times ( s^{s-1}p^s)=(sp)^s$. So we are done with (ii). 

To ensure (i), we note that according to the Bertrand-Chebyshev theorem there is always a prime number between $N$ and $2N$, for $N \geq 2$. 
Let $p$ be a prime between $n^2$ and $2n^2$, then $A(s,p)$ has more than $n^2$ elements and (i) is satisfied.

Note that since $p\leq 2 n^2$, our number of outputs satisfies 

	\begin{equation}
	\label{dbound1}
	d\leq (2sn^2)^s .
	\end{equation}
Having our set $A(s,p)$ satisfying (i) and (ii)
we now consider a game $M$ consisting  of elements in the form $\omega^{k_{ij}}$, where $k_{ij}\in A$ and $\omega=e^{i2\pi/d}$.

We shall now derive a bound on the classical value of the game constructed above using the girth method. To this end, we turn to the graph $H_{opt}$ 
corresponding to such a game. Due to the properties of the set $A$,
the graph has girth $g>2s$ i.e. it 
does not have cycles of length up to $2s$. Now, we need to estimate 
the number of edges of the graph $H_{opt}$, which in turn, according to \eqref{eq:cl-edges}, will determine the  classical value.

Erd\H os conjectured in 1964, and it was proved in 1974 in \cite{Bondy-girth},	that any $N$-vertex graph which  does not have a cycle of length less than or equal to $2s$ must have $m\leq O(N^{1+1/s})$ edges. 
Thus, we have a bound on the number $m$ of edges in $H_{opt}$ 
\begin{equation}
\label{eq:erdos}
m\lesssim (2n)^{1 + 1/s}
\end{equation}

A more precise result concerning graphs from  \cite{Hoory} is the following. For a bipartite graph with $m$ edges, $n_A + n_B$ vertices and girth $g>2s$  (hence $g\geq 2(s+1)$) we have 

\begin{equation}
n_A\geq  \sum_{i=0}^{s}\left(d_A-1\right)^{\left\lceil i/2\right\rceil}\left(d_B-1\right)^{\left\lfloor i/2\right\rfloor},
\end{equation}
\begin{equation}
n_B\geq  \sum_{i=0}^{s}\left(d_B-1\right)^{\left\lceil i/2\right\rceil}\left(d_A-1\right)^{\left\lfloor i/2\right\rfloor}, 
\end{equation}
where $d_A=\frac{m}{n_A}$ and $d_B=\frac{m}{n_B}$  are the {\sl average degrees} in the respective parts of the graph. 
For $n_A=n_B=n$ and $d_A=d_B=\frac{m}{n}$ this gives
\begin{equation} \label{eq:max-girth} 
n\geq  \sum_{i=0}^{s}\left(\frac{m}{n}-1\right)^i
\end{equation}

In particular, $\left(\frac{m}{n}-1\right)^s \leq n$, and solving this inequality for $m$ we obtain
\begin{equation} \label{eq:max-girth-simple} 
m\leq n+ n^{1+\frac{1}{s}} < 2n^{1+\frac{1}{s}},
\end{equation} 

\noindent which gives us a bound for  classical value 
$p_{Cl}\leq 2/n^{1-\frac{1}{s}}$, as needed.

Finally, if we set $s=n$ in equation \eqref{dbound1}, we obtain a bound on the number of outputs needed to maximize the number of contradictions.
\begin{equation}
\label{dbound2}
d\leq 2^n n^{3n}
\end{equation}

In this case, i.e. for $s=n$,
we can actually provide exact expression for the classical value. Indeed,
recall that for our matrix constructed based on set $A$, the graph $H_{opt}$ does not have cycles of length up to $2s$. On the other hand the graph has  $2n$ vertices and is connected (by Prop. \ref{prop:connected}).
Thus for $s=n$ the graph  has no cycles, and therefore it is a tree, which has $m=2n-1$ edges. That gives us the exact classical value of $p_{Cl}=\frac{2n-1}{n^2}$.
\end{proof}

\medskip
Again, the smallest values of $m$ that are of interest are $m=2n-1$,  when the graph is a tree with $g=\infty$, and $m=2n$, when the graph is a cycle of length $g=2n$. (In both of these cases 
\eqref{eq:max-girth}  becomes an equality.)  Next, the permitted values of $g$ decrease quickly as $m$ increases. For example, if $m=3n$ (average degree $d_A=d_B=3$), \eqref{eq:max-girth} implies 
$g \leq 2 \log_2(n+1)$; this bound is (approximately) saturated by Ramanujan expander graphs \cite{LPS}, which can be exhibited explicitly.

For $s=3$ (i.e. if the graphs do not contain cycles of length $6$ or smaller), \eqref{eq:max-girth-simple} 
reduces to $m\leq 2n^{4/3}$ and so the classical winning probability is upper-bounded as follows 
%
\begin{equation}
p_{Cl} = \frac{m}{n^2} \leq \frac{2}{n^{2/3}}.  
\end{equation}

\subsection{Prospects towards unbounded quantum  violation.}
\label{subsec:biases}
For nonlocal games, the most important feature is how much the quantum value exceeds the classical one. As an indicator we can use the following ratio of biases: 
\begin{equation}
R=\frac{p_Q-p_{rand}}{p_{Cl}-p_{rand}}
\end{equation}
where $p_{rand}$ is a random strategy (i.e. Alice and Bob provide always a random output).
We say that there is an unbounded violation if $R$ grows to infinity for growing number of inputs and outputs.

As already mentioned, there are very few explicit examples of games which exhibit unbounded violation.  Let us now argue that our results give a hope to obtain a new class of unbounded violation with unique features.  Namely, in Ref. \cite{RAM} there is the following  upper bound on quantum winning probability for linear games:
\begin{equation} \label{quantum_bound} 
p_Q\leq \bar p_Q\equiv\frac1d\left(1+ \frac{1}{\sqrt{n_An_B}}\sum_{k=1}^{d-1}\|M_k\|\right)
\end{equation}
where $M_k$ is entrywise $k$-th power, i.e. $(M_k)_{ij}=(m_{ij})^k$, 
and $\|\cdot\|$ is the operator norm. We consider now the case $n_A=n_B=n$. 
Using the fact that operator norm is no smaller than maximal norms of the columns, which is $\sqrt{n}$ we obtain that $\bar p_Q\geq \frac{1}{\sqrt{n}}$. 
We shall now suppose that, for a given game, the upper bound \eqref{quantum_bound} on quantum winning probability is saturated. 
In such case the ratio of the biases $R$ 
(and the ratio of probabilities) would satisfy
\begin{equation} \label{ratio1}
R=\frac{\overline{p}_Q -\frac1d}{p_{Cl}-\frac1d}\geq  \frac{\overline{p}_Q}{p_{Cl}} \geq
\frac {n^{-\frac12}}{2n^{-\frac{2}{3}}}= \frac{n^{1/6}}{2} .
\end{equation}

We can therefore hope for $R = \Omega( n^{1/6})$ with the number of outputs $d\sim O(n^6)$ (as given by \eqref{dbound1}). 

Note that although $R$ is unbounded, both the quantum value and the classical value go to zero with growing number of inputs and outputs. In this regard our game is like Khot-Vishnoy game 
\cite{Khot2005,Kempe2008a}, 
and should be contrasted with e.g. Hidden Matching game of \cite{Buhrman} where the quantum value stays constant. This means that in our case, one would need to repeat the experiment approximately $\sqrt{n}$ times to demonstrate quantum advantage.

For a general $s$, one may hope to achieve (by the same argument)  
\begin{equation} \label{ratio2}
R = \Omega\big( n^{\frac12 - \frac1s}\big) .
\end{equation}
%
In particular, if we take $s=\Theta(\log(n))$, the ratio $R$ is nearly  $n^{1/2}$. 

While, for small $n$,  it is possible to achieve the maximum contradiction number with a significantly smaller number of outputs, it is conceivable that 
the above scaling  (i.e., $\log d \sim n \log n$) is asymptotically optimal. 

As a matter of fact, if the bound of \cite{RAM} was saturated for our games (or even approximately saturated), one would obtain unbounded violation for a regime, which is not covered in existing explicit constructions.
Indeed, one can measure the efficiency of obtaining large violation  by considering $R$ as a function of $\io=\# \text{ inputs} \times \# \text{ outputs}$. 

Table \ref{tab:io} shows comparison  of the existing results on explicit constructions with the unbounded violation that we are aware of with those potentially offered by our results.
\begin{table*}
\begin{center}
	\begin{tabular}{|c|c|c|c|}
		\hline
		&  parallel repetition techniques& Khot-Vishnoi game & our (provided saturation  holds)\\ 
		\hline\hline
		$R(\io) \gtrsim $   & $\io^{10^{-5}/2}$ & $ \frac{\log(\io)}{[\log\log(\io)]^2} $ & $\io^{1/10}$\\
		\hline
	\end{tabular}
\end{center}
	\caption{\label{tab:io} Violation as a function of product of the number of inputs and outputs (denoted by ``io''). The best ratio for non-constructive games is 
	$R(\io)=\io^{1/4}/\log(\io)$ of \cite{Junge2011}. The ratio based on parallel repetition techniques is taken from \cite{Junge2010}.}
\end{table*}
Here we give value of  
Khot-Vishnoi game 
\cite{Khot2005,Kempe2008a} in the version from \cite{Buhrman}
(the construction of 
 \cite{Regev} is slightly worse according to criterion indicated in the previous paragraph). By the ``parallel repetition techniques'' 
 we mean the games that originate from playing a given game many times in parallel (see  \cite{Raz-parallel}, \cite{Rao}, \cite{Junge2010}, \cite{single-shot-Araujo}).  

We hope therefore that our results will stimulate the search for the quantum value of linear games, especially of the type constructed by  us. {\cred We wish to point out that while expecting 
the bound \eqref{quantum_bound} to be saturated/approximately saturated may be overly optimistic, the minoration $p_Q\geq \frac{1}{\sqrt{n}}$, which is needed for \eqref{ratio1} or \eqref{ratio2}, is quite conceivable. Indeed, we use a lower bound $\|M_k\| \geq \sqrt{n}$, which -- in our context -- is tight only if $M_k$ is unitary. On the other hand, 
the true value of $p_Q$ -- as analyzed in \cite{RAM} -- involves (vaguely) expressions of the form $\langle \alpha_k |M_k|\beta_k\rangle$ (for some appropriate unit vectors 
$\alpha_k,\beta_k$) rather than $\|M_k\|$, and when $M_k$ is unitary, there are plenty of choices of $\alpha_k,\beta_k$, for which $\langle \alpha_k |M_k|\beta_k\rangle = \|M_k\|$. While this argument is far from being a proof, it does show that there is some ``wiggle room.''}

\subsection{Lower bounds on the number of outputs}
Both methods described above provide explicit constructions of matrices with maximum contradictions and reasonably low numbers of outputs. However, neither of them is optimal. 
It is, in fact, possible to achieve the maximum number of contradictions with a much fewer outputs. For example, the matrix

\begin{align}
\label{eq:3x3-beta-4}
\left(\begin{array}{ccc}
1 & 1        & 1\\ 
1 & \omega   & \omega^{3}\\
1 & \omega^4 & \omega^5
\end{array}\right)    
\end{align}

\vspace{2ex}

\noindent with $\textit{d = 7}$ outputs has $4$ contradictions, see Section \ref{3x3}, Table \ref{tab:3x3}. 
We later show (see Proposition \ref{prop:chromatic}) that $\textit{7}$ is the minimum number of outputs necessary to achieve maximum  number of contradictions  in a $3\times 3$ matrix. The exact values of $d_{min}$ for larger matrix sizes remain unknown. We do however provide a lower bound on the number of outputs. To this end we construct a graph $G_n$ such that $d$ can be bounded from below in terms of its chromatic number (see \cite{BondyMurty}).

Every set $X$ within a matrix which corresponds to a cycle in $H_{opt}$ is defined by a pair of permutations $\pi_1,\pi_2 \in S_n$. The set $X$ consists of elements $x_{i\pi_j(i)}$ for $i\in [n], j\in\{1,2\}$. If $\pi_1(i)=\pi_2(i)$ for some $i\in [n]$, the element $x_{i\pi_j(i)}$ is removed from $X$.

It follows that if for each pair of permutations $\pi_1,\pi_2 \in S_n$ we have

\begin{center}
\begin{equation}
\label{permut}
\sum\limits_{j=0}^{n-1} k_{j,\pi_1(j)} \neq \sum\limits_{j=0}^{n-1} k_{j,\pi_2(j)} \mod d,
\end{equation}
\end{center}

\noindent then the graph contains no cycles without contradictions, and thus $\beta_C=(n-1)^2.$ 

However, not every such pairs of permutations defines exactly one cycle in $K_{n,n}$. For a set $\pi_1,\pi_2$ defining two or more disjoint cycles it is possible that (\ref{permut}) holds but no cycle without contradiction exists.
\begin{proposition}
\label{prop:chromatic}
The number $d$ of outputs necessary to achieve the maximum number of contradictions in an $n\times n$ matrix is at least the chromatic number of the graph $G_n$, in which the vertices are all permutations $\pi \in S_n$ and two vertices $\pi_1,\pi_2$ are adjacent if and only if they define exactly one cycle in $K_{n,n},$ i.e. $\pi=\pi_1\pi_2^{-1}$ is a cyclic permutation.

That in turn is bounded from below by the maximum of:
\begin{enumerate}
\item $\left|C\right|+1$, where $C\subset S_n$ is the largest set of cyclic permutations such that $\pi_i\pi_j^{-1}$ is a cycle for any $\pi_i, \pi_j\in C$.
\item $\frac{n!}{\left|J\right|},$ where $J$ is the largest set of cyclic permutations such that $\pi_i\pi_j^{-1}$ is not cyclic for any $\pi_i,\pi_j\in J.$
\end{enumerate}
\end{proposition}

\begin{center}
\begin{figure}[H]
	\begin{picture}(300,140)
	\put(0,-10){\includegraphics[scale=0.30]{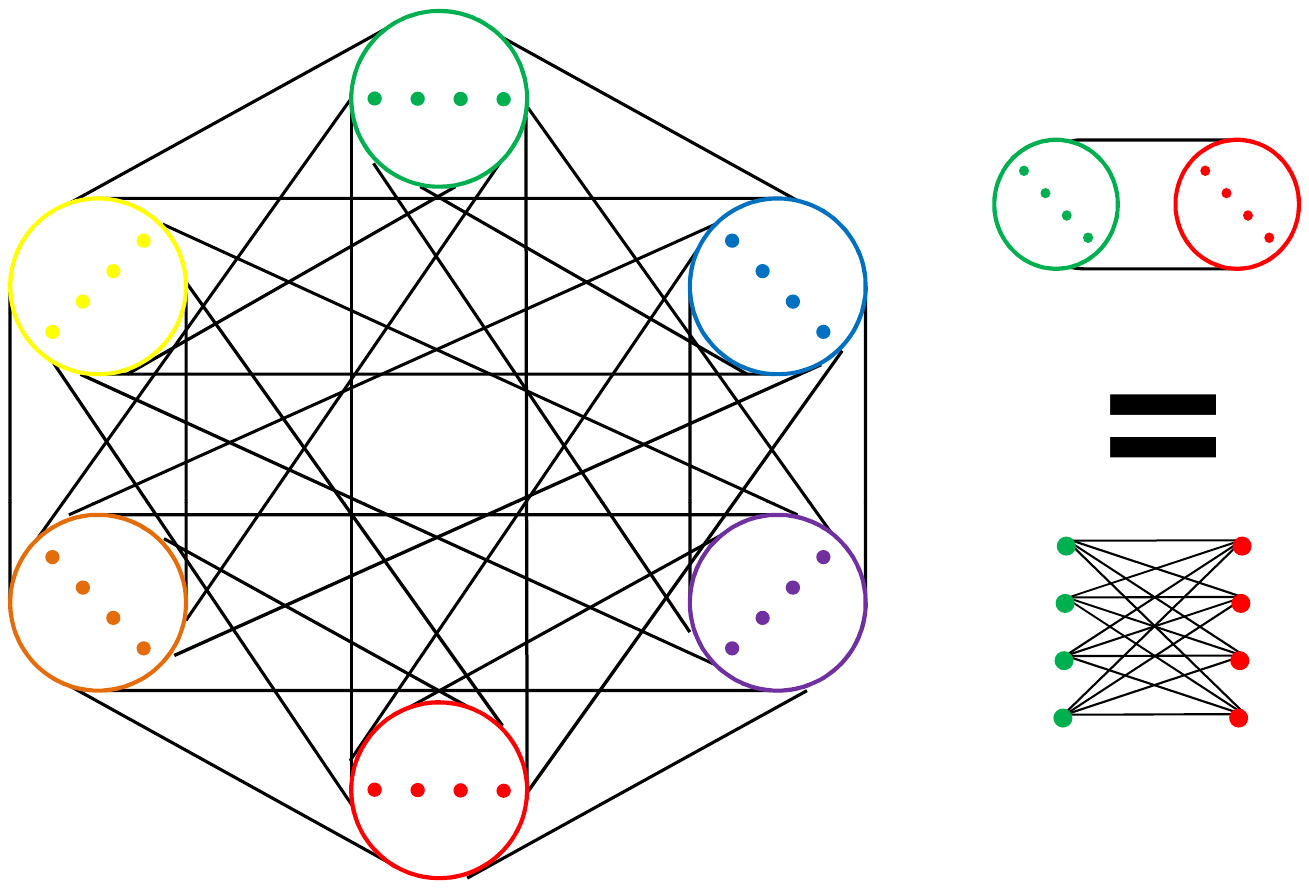}}
	\end{picture}
	\caption{The graph $G_n$ for $n=4$ is $K_{4,4,4,4,4,4}$. In the above figure each colored circle represents an independent set of four vertices. A connection between two such sets means that every vertex from one set is adjacent to every vertex from the other.
	\\$G_3$ is simply the complete graph $K_6.$ For $n\geq 5$ the structure becomes more complicated.}
	\label{fig:G4}
\end{figure}
\end{center}

The chromatic number of the graph $G_n$ in Proposition \ref{prop:chromatic} can also be bounded in terms of eigenvalues of its adjacency matrix $A(G)$. For example in \cite{Hoffman} it is shown that $\chi_{G}\geq 1+\frac{\lambda_1}{\left|\lambda_n\right|}$, where $\lambda_1$ and $\lambda_n$ are the largest and smallest eigenvalues of $A(G).$ Since $G_n$ is a regular graph, it can be shown that $\lambda_1$ is equal to the degree of the vertices, i.e. $\lambda_1=\sum_{i=2}^{n}\frac{n!}{(n-i)!i}$.

For $n=3$ the graph $G_n$ is the complete graph $K_6$. Its chromatic number is $6$. However, by brute force check using Table \ref{tab:3x3}, one finds that for $n=3$ and $d=6$ outputs, one gets at most 3 contradictions. \blk
In the case of the graph $G_4$ in Fig. \ref{fig:G4}, the chromatic number is still $6$. We believe that for larger values of $n$ the chromatic number of $G_n$ increases, providing a reasonable bound on $d$, although we do not yet know whether this bound can achieved.

\section{Connection to the labeled graph framework}
\label{sec:graphs}

In this section we compare the matrix-based approach to the framework described in \cite{RS} and \cite{RRGHS}, in which linear games 
 
are described in terms of graphs with permutations assigned to the edges. In many ways the two frameworks are similar. A labeled graph can be easily translated into a matrix and vice-versa. However, each approach has its own strengths and weaknesses. Some properties of these games are easier to prove in graph-theoretic terms, others are more evident in a matrix.  

Actually, the labeled graph approach is suitable for a slightly more general family of games,  called ``unique games.''  
A~unique game is a game in which the variables can take values from the set $[n]=\{0,...,n-1\}$ and the constraints are in the form $\pi_{xy}(a)=b,$ where $a$ and $b$ are values assigned to $x$ and $y,$ respectively and each  $\pi_{xy}$ is a permutation. Thus we shall recall the approach involving  labeled graphs in the more general context of unique games, and then will specialize to the case of linear games.

In subsection \ref{3x3} we provide, as an example, an systematic analysis of games with three inputs on each side, using a combination of graph and matrix tools.

\subsection{Unique games as labeled graphs}
\label{sec:unique}

A unique game can be described in terms of a labeled graph $(G,K)$, where $G$ is a graph with vertex set $V$ and edge set $E$, and the edge-labeling $K:E\mapsto S_n$ assigns a permutation of the set $[n]=\{0,1,...,n-1\}$ to each edge of $G$. Typically, $G$ is a directed graph, but it may be undirected if all permutations used are such that $\pi^{-1}=\pi$. In a bipartite graph, we are going to assume a left to right default orientation for all edges.

A classical strategy for a game defined in these terms corresponds to a vertex-assignment $f:V\mapsto [n].$ By a \textit{contradiction} in a given assignment $f$ we mean an edge $xy\in E$ such that $\pi_{xy}(f(x))\neq f(y).$ The \textit{contradiction number} $\beta_C$ is the minimum number of contradictions over all possible assignments. The classical value of a game represented by $(G,K)$ can be expressed as $p_{Cl}=1-\frac{\beta_C(G,K)}{\left|E(G)\right|}$ We consider two labeled graphs to be equivalent if one can be obtained from the other through the following operations: 
\begin{enumerate}
	\item Isomorphism between the underlying (unlabeled) graphs
	
	\item In a directed graph, replacing an edge $\overrightarrow{xy}$ labeled with $\pi$ with $\overrightarrow{yx}$ labeled with $\pi^{-1}.$
	
	\item Switching operations $s(x,\sigma)$ for any vertex $x\in V(G)$ and permutation $\sigma\in S_n,$ defined as follows. For every vertex $y\in N_{G}(x)$:
	
	\begin{enumerate}
		\item if $\overrightarrow{yx}\in E(G),$ we replace $K(\overrightarrow{yx})=\pi$ with $K'(\overrightarrow{yx})=\sigma\pi$,
		
		\item if $\overrightarrow{xy}\in E(G),$ we replace $K(\overrightarrow{xy})=\pi$ with $K'(\overrightarrow{xy})=\pi\sigma^{-1}$.
	\end{enumerate}
\end{enumerate}

Since the equivalence relation preserves the~contradiction number of a labeled graph, it is clear that games defined by equivalent labeled graphs have the same classical value. 
The same is true about their quantum values since graph isomorphisms and switches represent the corresponding relabeling of the inputs $x,y$ and outputs $a,b$, respectively. Thus, we only need to calculate the value for one labeled graph in each equivalence class. Hence unique games defined by equivalent labeled graphs are considered equivalent.

\subsection{Linear games:  labeled graphs versus matrices }

 The complete bipartite graph $K_{n,n}$ labeled with permutations $\sigma_{xy}(a)=a+k_{xy} \mod d$, defines the same type of game introduced in Section \ref{sec:est}, with the predicate given by equation \eqref{eq:V}.
Thus, we can also describe this game in terms of a matrix, in which all elements are complex roots of $1$.  Notice that in this case equivalent labeled graphs correspond to equivalent matrices. The matrix operations defined in Section \ref{sec:est} are essentially the same as the graph operations in Section \ref{sec:unique}. The first one is a switch on a vertex. The second one is an isomorphism changing the order of vertices on one side of the graph. The third one swaps the two sides. 

This connection allows us to examine the games in question from two different angles, applying both graph and matrix tools to study their properties. 

It follows from the results in \cite{RS} that for any bipartite graph $G$ all labelings $K:E\mapsto L_d'=\{\tilde{\sigma_i}\in S_d: \tilde{\sigma_i}(a)=i+a \mod d\}$ with no contradiction are equivalent to the labeling $K_0$, where $K_0(e)=\id$ for all $e\in E$. In terms of matrices this means that every  game with no contradiction is equivalent to a matrix in which all elements are equal to 1, i.e. a matrix of rank 1. Thus, the contradiction number of a game, as defined in terms of a labeled graph, is equal to the rigidity ${\rm Rig}(M,1)$ of the corresponding matrix.

It is easy to see that contradictions within a labeled graph arise from cycles. Furthermore, it is shown in \cite{RS} and \cite{RRGHS}  that a graph labeled with $L_d'$ has a contradiction if and only if it contains a bad cycle, i.e., a cycle with a contradiction. Every cycle in the labeled graph $(G,K)$ corresponds to a cycle in the matrix $M$. Cycles which do not give rise to contradictions are called good cycles and correspond directly to good cycles in a matrix, i.e., 
those satisfying Eq. \eqref{eq:good-cycles-matrix}. 
More precisely, a good cycle in $(G,K)$ corresponds to a cycle in the corresponding graph $H_{opt}$ and, as such, can be defined by a set of matrix elements satisfying the four conditions in Section \ref{sub:max}. A cycle containing a contradiction corresponds to a set satisfying conditions 1-3, but not 4. 
Accordingly,  in what follows we will consider cycles within the labeled graph.

However, we do not need to consider all cycles in $(G,K)$. It follows from the results of \cite{RS} that a complete bipartite graph in which no cycle of length $4$ contains a contradiction, cannot in fact contain any contradiction. Thus we only need to consider cycles of length $4$, which correspond to $2 \times 2$ submatrices of a matrix.

\subsection{Labeled graph representation versus the graphs $H$ and $H_{opt}$}
\label{sub:label-versus-H}

Notice that the graphs $H$ and $H_{opt}$  introduced in Section \ref{sec:girth-method}  can be alternatively defined as subgraphs of the graph $G$ in $(G,K)$. The vertex set of these graphs is the same as that of $G$. The edge set of $H$ is the set of edges of $G$ labeled with $\id$. The graph $H_{opt}$ is a subgraph whose edge set is a maximal set of edges containing no contradiction. It follows that every cycle in $H_{opt}$ corresponds to a good cycle in the labeled graph $(G,K)$, and vice versa.

Some properties of a game can be inferred from $H$ or $H_{opt}$ alone. 
 However, the labeled graph $(G,K)$ contains more information about the game and  
{\it prima facie} it 
can be used to show things that are not evident from $H_{opt}$. 

\smallskip 
As an illustration of the relationship between the two approaches, we will state a fact which is a counterpart of 
Corollary \ref{cor:good-cycle} and prove it within the labeled graph framework. 
\begin{lem}
\label{lcycles}
A complete bipartite labeled graph $(K_{kl}, K:E\mapsto L_d')$ has the maximum number of contradictions iff every cycle in the graph contains a contradiction.

\end{lem}

\begin{proof}
It is clear that a graph with no cycles cannot contain a contradiction. Let $C$ be a cycle in $K_{kl}$ with no contradiction. We define the subgraph $G$ of $K_{kl}$, where $V(G)=V(K_{kl})$ and $E(G)$ is obtained by taking all edges of $C$ and adding as many other edges from $E(K_{kl})$ as possible without creating another cycle in $G.$ Clearly, the graph $G$ has no contradiction and $\left|E(G)\right|=k+l$, as contradictions can only occur within cycles. It follows that $K_{kl}$ has no more than $kl-(k+l) = (k-1)(l-1)-1$ contradictions.  We have already showed in Section \ref{sec:girth-method} that the maximum contradiction number is $(k-1)(l-1)$ Thus the maximum number of contradictions can only be achieved if there is a contradiction in every cycle.

Now we show that $\beta_C<(k-1)(l-1)$ implies the existence of a cycle with no contradiction. If $\beta_C(K_{kl},K)$ is less than the maximum, then there exists a labeling $K'$ of $K_{kl}$ such that $(K_{kl},K')$ is equivalent to $(K_{kl},K)$ and $K'(e)=\id$ for at least $k+l$ edges $e\in E(K_{kl})$. But a graph with $k+l$ vertices and $k+l$ edges must contain a cycle. Thus there is a cycle with no contradiction in $(K_{kl},K')$, which implies the existence of a cycle with no contradiction in $(K_{kl},K)$. Therefore, if $(K_{kl},K)$ has no cycle without contradiction, then $\beta_C(K_{kl},K)=(k-1)(l-1)$.
\end{proof}

Thus every $k\times l$ matrix with $(k-1)(l-1)$ contradictions corresponds to a complete bipartite graph $K_{kl}$ in which every cycle contains a contradiction, and vice versa. The graph $H_{opt}$ corresponding to such a matrix is a tree (cf. Fact \ref{fact:max-contradiction-tree}).

The labeled graph framework can be applied to a wider variety of quantum systems than the matrix framework described in this paper. For example this type of matrix corresponds only to complete bipartite graphs. However, the matrix approach can be modified to describe  a wider variety of games.  A graph labeled with a different set of permutations than $L_d'$ can also be described in terms of a matrix with elements from a group other than the complex roots of $1$. The matrix approach can also be extended to non-bipartite graphs, as every graph can be represented as an adjacency matrix. A formalism based on adjacency matrices would require us to modify our approach somewhat and, in the specific case of complete bipartite graphs, it would be unnecessarily complicated, but it can be a useful tool for calculations based on the graph framework.

\subsection{Detailed analysis of the case $n=3$} 
\label{3x3}

We will now combine matrix and graph methods to study the contradiction number of a game defined by a $3\times 3$ matrix
\begin{align}
\label{eq:general3x3}
M=\left(\begin{array}{ccc}
1 & 1 & 1\\ 
1 & w & x\\
1 & y & z
\end{array}\right).    
\end{align}

\vspace{2ex}
The main results of this section are contained in Tables \ref{tab:3x3}  and \ref{tab:3x3-ones}, where we provide the contradiction number of $M$ for arbitrary values of the above four parameters $w,x,y,z$.
In Table 
\ref{tab:3x3} we present the cases where
none of the numbers $x,y,z,w$ are equal to $1$, while in Table \ref{tab:3x3-ones} we treat the cases where some of them are $1$s. 
An example of a matrix with the maximal number of contradictions (i.e. $\beta_C=4$) and minimal number of outputs ($d=7$)  is given by Eq. \eqref{eq:3x3-beta-4}.

\smallskip 
We will first prove the following proposition, which determines the contradiction numbers of $3\times 3$ diagonal matrices, the size that was not addressed in Proposition \ref{diag}.

\begin{proposition}
\label{prop:diag3by3}
For a $3\times 3$ matrix 
%
\begin{align}
\label{eq:diagonal3x3}
M=\left(\begin{array}{ccc}
x_0 & 1 & 1\\ 
1 & x_1 & 1\\
1 & 1 & x_2
\end{array}\right). 
\end{align}
If $x_i\neq 1$ for $i=1,2,3$, then we have \\
{\rm (i)}  $\beta_C(M)= 2$ if $x_1=x_2=x_0^{-1}$ or $x_0=x_2=x_1^{-1}$
or 
$x_0=x_1=x_2^{-1}$\\
{\rm (ii)} $\beta_C(M)= 3$ otherwise.\\
If some of the diagonal entries are equal to $1$, $\beta_C(M)$ equals to the number of elements different from $1$.

\end{proposition} 
Note that since every $3\times 3$ matrix which has at most one element different from $1$ in each row and each column is equivalent to a matrix of form \eqref{eq:diagonal3x3}, Proposition \ref{prop:diag3by3} allows to determine contradiction numbers of all such matrices.

\smallskip
\begin{proof}
Consider a $3\times 3$ matrix $M$ of the form \eqref{eq:diagonal3x3} 
 and let $M'$ be equivalent to $M$. We first observe that  all elements of $M'$ that are different from $1$ can not be located in the same row (or in the same column). Indeed, if that was the case, then $M'$ would contain a $2\times 3$ (or $3\times 2$) submatrix of rank $1$, while all  $2\times 3$ or $3\times 2$ submatrices of $M$ are clearly of rank $2$, which would yield a contradiction in view of Remark \ref{minor2minor}.

Consequently, every $M'$ equivalent to $M$ must contain at least two entries different from $1$ (hence $\beta_C(M) \geq 2$)  and if it contains exactly two such entries, then those two entries must be in different columns and different rows. 

With the above observation in mind, let us investigate the constraints imposed on $M$ by the condition $\beta_C(M) = 2$.  As an illustration, suppose that a matrix of the form 
\begin{equation}\label{2elements}
M'= (m'_{ij})_{i,j=0}^2 = \left(\begin{array}{ccc}
1 & 1 & 1\\ 
1 & 1 & y_2\\
1 & y_1 & 1
\end{array}\right)
\end{equation}
is equivalent to $M$. We note that $M'$ has (at least) two $2\times 2$ submatrices that consist solely of $1$s, namely those given by $i,j \in \{0,1\}$ and by $i,j\in \{0,2\}$. By Remark \ref{minor2minor}, the same must be true for $M$. Now, every $2\times 2$ submatrix of $M$ includes at least one diagonal element $x_i$ and if it includes exactly one, then the corresponding minor can not be zero. So only submatrices containing two diagonal elements may lead to a zero minor.  If those two minors are given by $i,j \in \{0,1\}$ and by $i,j\in \{0,2\}$, the requirements for them to be zero are respectively 
\begin{equation} \label{conditions}
x_0x_1-1=0 \  \hbox{ and } \ x_0x_2-1=0,
\end{equation}
which is equivalent to $x_1=x_2=x_0^{-1}$. By symmetry, the other two sets of constraints from (i) are related to the remaining choices of two pairs of indices from among $\{0,1\}$, $\{0,2\}$ and $\{1,2\}$. 

Conversely, if the  
constraints \eqref{conditions} are satisfied, then multiplying 
the first row of $M$ by $x_0^{-1}$ and then the second and the third column by $x_0$ we obtain 
\begin{equation} \label{suffice}
M' = \left(\begin{array}{ccc}
1 & 1 & 1\\ 
1 & x_0x_1 & x_0\\
1 & x_0 & x_0x_2
\end{array}\right)= 
\left(\begin{array}{ccc}
1 & 1 & 1\\ 
1 & 1 & x_0\\
1 & x_0 & 1
\end{array}\right),
\end{equation}
which shows that $\beta_C(M')=\beta_C(M) \leq 2$, and hence both are equal to $2$. 

The same argument shows that any matrix of form \eqref{2elements} with $y_1, y_2\neq 1$, or equivalent to it,  satisfies $\beta_C(M')=2$. Together with Proposition \ref{prop:one_row} (or Remark \ref{minor2minor}), this also justifies the last assertion of the Proposition.  
\end{proof}

\vspace{2ex}

Now let us return to the contradiction number of the matrix \eqref{eq:general3x3}.
To begin with, note that the number of contradictions in this matrix is at most the number of elements $w,x,y,z$ different from 1. For the time being, let us assume that $1\notin\{w,x,y,z\}$.

It is easy to see that $\beta_C = 1$ if $w=x=y=z,$ as multiplying the last two rows by $x^{-1}$ and the first column by $x$ transforms it into 

\begin{align}
M'=\left(\begin{array}{ccc}
x & 1 & 1\\ 
1 & 1 & 1\\
1 & 1 & 1
\end{array}\right).    
\end{align}
If any three of the elements are equal, say $w=x=y,$
and the fourth element is different, 
multiplying the last two rows by $x^{-1}$ and the first column by $x$ transforms it into 
\begin{align}
M'=\left(\begin{array}{ccc}
x & 1 & 1\\ 
1 & 1 & 1\\
1 & 1 & zx^{-1}
\end{array}\right).
\end{align}
The corresponding labeled graph has multiple bad 4-cycles and no single edge belonging to all of them. Thus, $\beta_C=2.$ (This also follows from Proposition \ref{prop:diag3by3}.) 

Consider next the case of exactly one equality among $w,x,y,z$.  If the two equal elements are in the same row or column, we can easily reduce the number of non-1's to three by multiplying that row or column by the appropriate factor. For example, if $w=x$, we multiply the second row by $x^{-1}$.  The resulting matrix 
\begin{align}
M'=\left(\begin{array}{ccc}
1 & 1 & 1\\ 
x^{-1} & 1 & 1\\
1 & y & z
\end{array}\right) 
\end{align}
has precisely one zero minor (this uses $x\neq y$, $x\neq z$) and so, by the argument from the proof of Proposition \ref{prop:diag3by3}, can not be equivalent to a matrix with two (or less) non-$1$ entries. It follows that, in that case, $\beta_C(M)=3$.

If the two equal elements are not in the same row or column, for example $w=z$, we can multiply the last two rows by $z^{-1}$ and the first column by $z$. Then we obtain the matrix. 

\begin{align}
M'=\left(\begin{array}{ccc}
z & 1 & 1\\ 
1 & 1 & xz^{-1}\\
1 & yz^{-1} & 1
\end{array}\right).    
\end{align}
To show that there is no equivalent matrix with fewer non-1 elements, we use  Proposition \ref{prop:diag3by3} and  it follows that the  original  matrix $M$ had the contradiction number equal to $3$. 

The above argument also applies if {\sl two pairs} of elements on the diagonals are equal, i.e., if $w=z$ and $x=y$, but $x\neq z$. This is because condition (i) of Proposition \ref{prop:diag3by3} would imply $xz^{-1}= z^{-1}$ and therefore cannot be satisfied; accordingly, the contradiction number is also equal to $3$. 

If $w=x$ and $y=z$, the matrix is clearly equivalent to 

\begin{equation}
M'=\left(\begin{array}{ccc}
1 & 1 & 1\\ 
x^{-1} & 1 & 1\\
z^{-1} & 1 & 1
\end{array}\right).
\end{equation}
\vspace{2ex}

If $x\neq z$ (and hence $x^{-1}\neq z^{-1}$), it follows from Proposition \ref{prop:one_row} that $\beta_C=2.$
Similarly, if $w=y$ and $x=z$, but $x\neq y$, then $\beta_C = 2$.

It follows that $\beta_C = 4$ is only achievable if no two elements in the set $\{w,x,y,z\}$ are equal. However, it is still possible that all these elements are different and $\beta_C=3.$ We shall analyze that situation using the labeled graph formalism. 

\begin{center}
\begin{figure}[H]
	\begin{picture}(300,140)
	\put(10,-10){\includegraphics[scale=0.30]{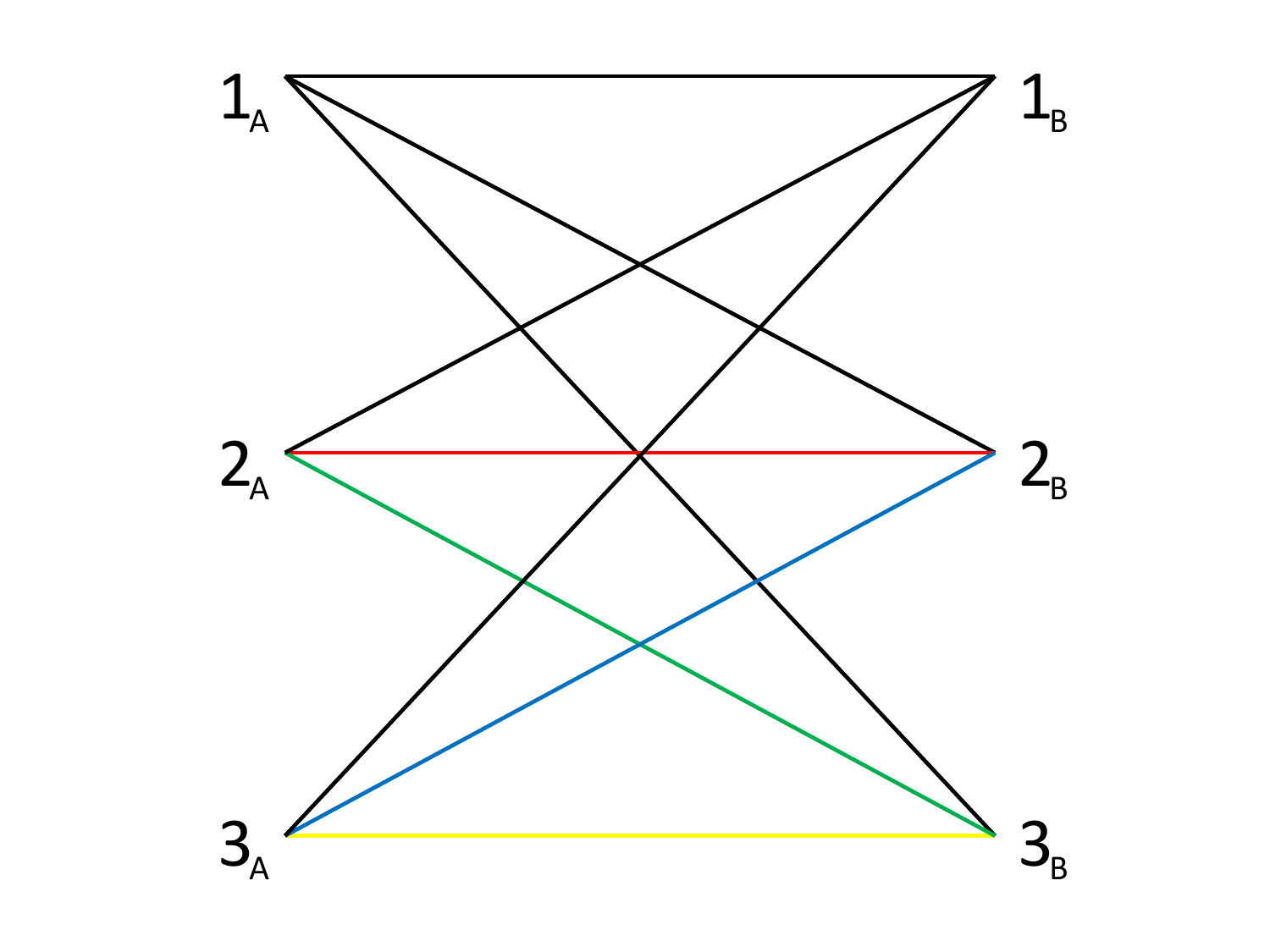}}
	\end{picture}
	\caption{Each edge color represents a different permutation from the set $L_d'=\{\tilde{\sigma_i}: \tilde{\sigma_i}(a)=i+a\}$.}
	\label{fig:K33}
\end{figure}
\end{center}

The labeled graph described by the matrix $M$ (see Fig. \ref{fig:K33}) contains fifteen overlapping cycles (see Fig. \ref{fig:C4} and \ref{fig:C6}).

\begin{center}
\begin{figure}[H]
	\begin{picture}(300,140)
	\put(10,-10){\includegraphics[scale=0.30]{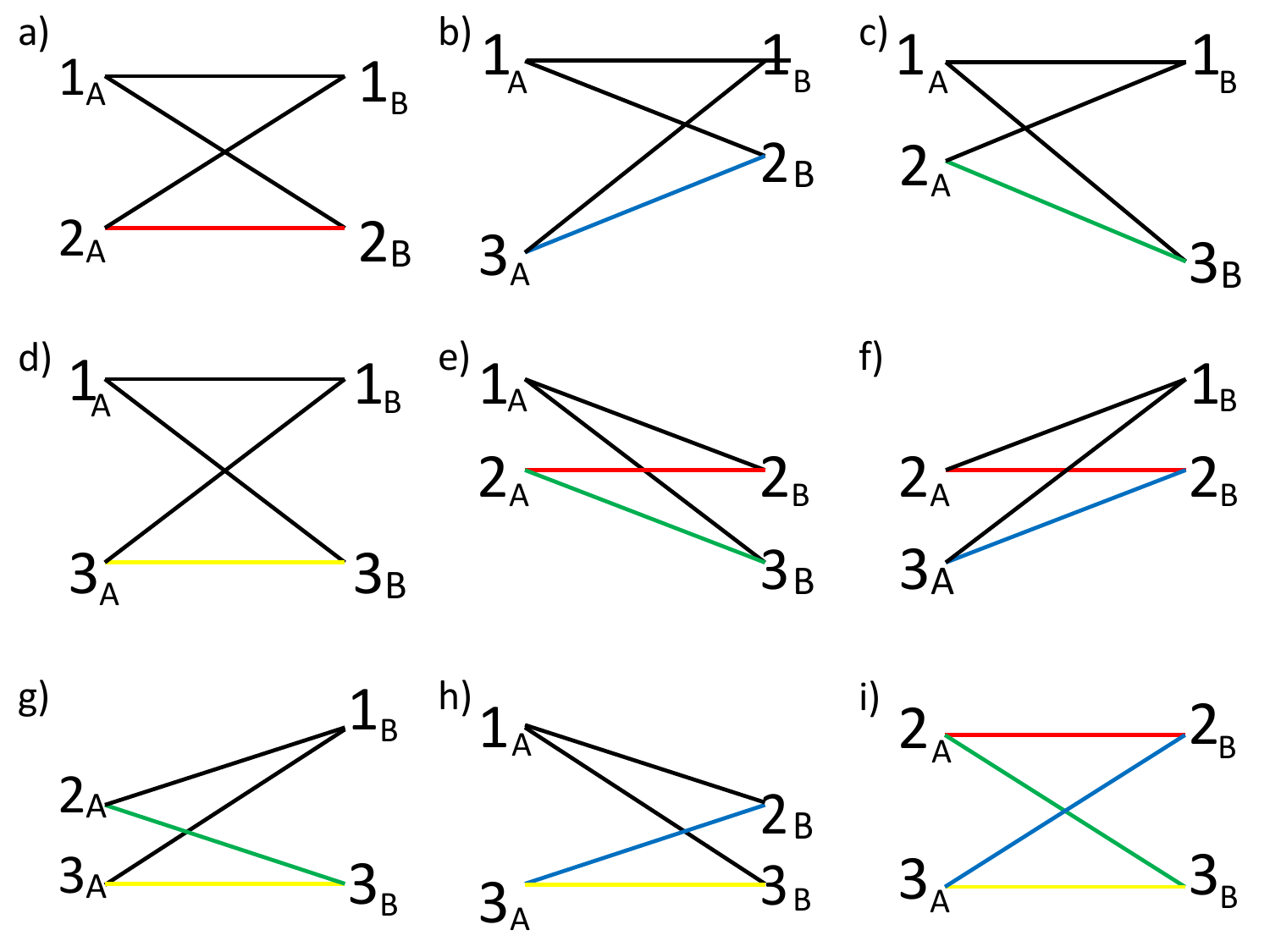}}
	\end{picture}
	\caption{Cycles of length 4 within the graph.}
	\label{fig:C4}
\end{figure}
\end{center}

If no two elements in the set $\{w,x,y,z\}$ are equal, then all cycles $a)-h)$ and $n)-o)$ contain contradictions. To get rid of all those cycles we must delete at least three edges. 

\begin{center}
\begin{figure}[H]
	\begin{picture}(300,140)
	\put(10,-10){\includegraphics[scale=0.30]{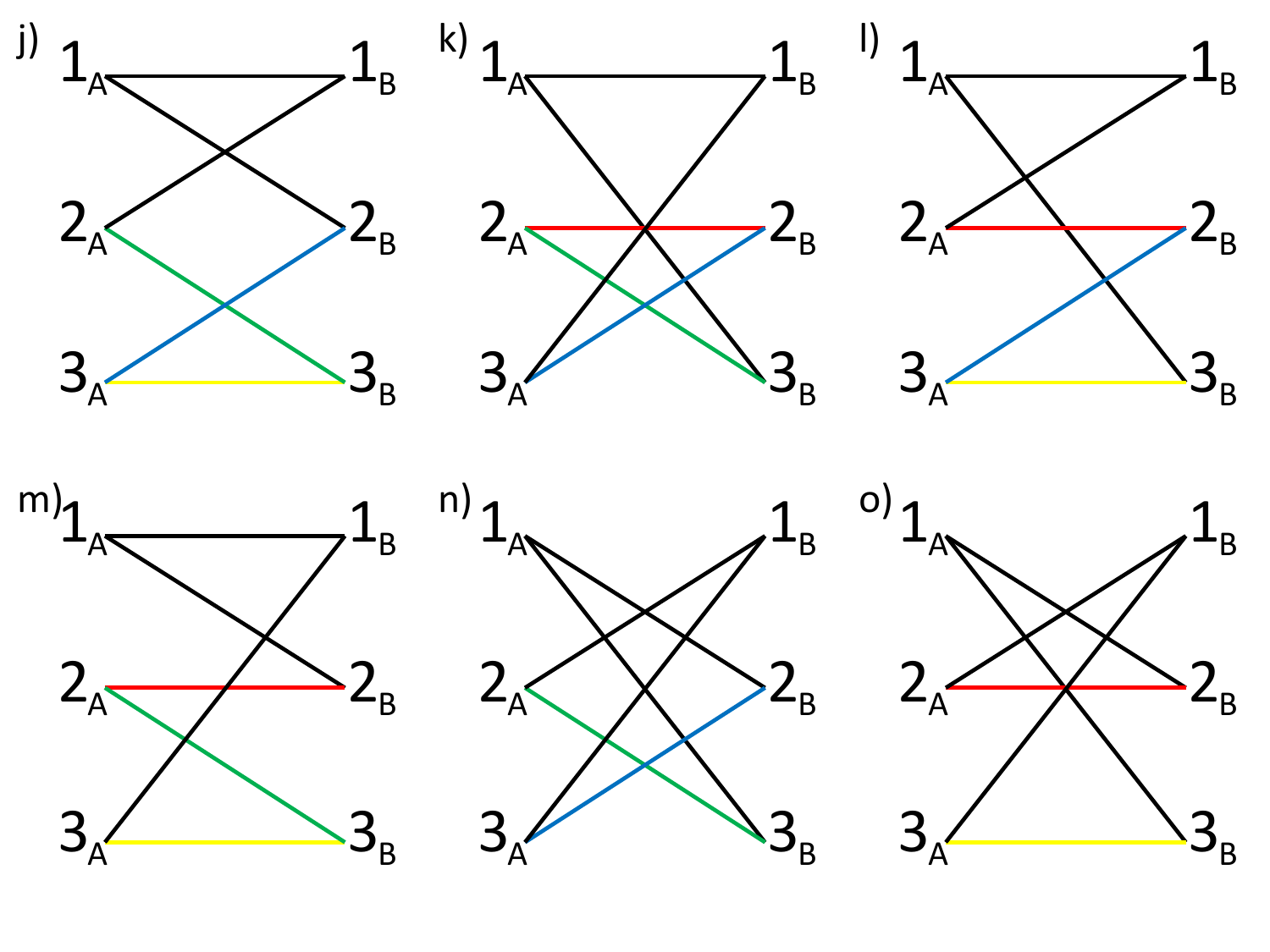}}
	\end{picture}
	\caption{Cycles of length 6 within the graph.}
	\label{fig:C6}
\end{figure}
\end{center}

Deleting any three edges such that the remaining graph does not contain any of the cycles $a)-h)$ and $n)-o)$ leaves us with exactly one of the cycles $i) - m)$. Thus, clearly, $\beta_C = 3$ if at least one of these cycles contains no contradiction, and $\beta_C=4$ otherwise.

Recall that the elements of the matrix correspond directly to permutations from the set $L_d'=\{\tilde{\sigma_i}\in S_d: \tilde{\sigma_i}(a)=i+a$ mod $d\}$ assigned to the edges of the graph. The number $\omega^i$ represents an edge labeled with the permutation $\tilde{\sigma_i}$. In particular, a $1$ represents the identity, or $\tilde{\sigma_0}.$ Thus, the conditions for specific cycles containing no contradictions can be written in terms of the values in the matrix for example $\sigma_w^{-1}\sigma_y\sigma_z^{-1}\sigma_x=\tilde{\sigma_0}$, the condition for the cycle i) being good, can also be written as $w^{-1}yz^{-1}x=1.$

\begin{table}
\begin{center}
\begin{tabular}{|c|c|}
\hline
\textbf{Values of $w,x,y,z$} & \textbf{Number of contradictions}
\\\hline\hline $w=x=y=z$ & $\beta_C = 1$
\\\hline any three values equal & $\beta_C = 2$
\\\hline $w=x$ and $y=z$ & $\beta_C = 2$
\\\hline $x=z$ and $w=y$ & $\beta_C = 2$
\\\hline $x=y$ and $w=z$ & $\beta_C = 3$
\\\hline any two values equal & $\beta_C = 3$
\\\hline all values different and&
\\$w^{-1}yz^{-1}x=1$&
\\or&
\\$yz^{-1}x=1$&
\\or&
\\$w^{-1}z^{-1}x=1$& $\beta_C = 3$
\\or&
\\$w^{-1}yx=1$&
\\or&
\\$w^{-1}yz^{-1}=1$&

\\\hline all values different, otherwise & $\beta_C = 4$
\\\hline
\end{tabular}
\end{center}
\caption{\label{tab:3x3} Characterization of $3\times3$ games \eqref{eq:general3x3} with respect to the contradiction number if $1 \not\in \{w,x,y,z\}$.}
\end{table}

\begin{table}
\begin{center}
\begin{tabular}{|c|c|}
\hline
\textbf{Values of $w,x,y,z$} & \textbf{Number of contradictions}
\\\hline\hline 
one non-$1$ & $\beta_C = 1$
\\\hline two equal non-$1$'s  & 
\\ in the same  & $\beta_C = 1$
\\ row or column &
\\ \hline
two  non-$1$'s, otherwise & 
$\beta_C=2$
\\ \hline
three non-$1$'s and &
\\ $x=y$ &
\\ or $x=z$ & $\beta_C=2$
\\ or $xy=z$ &
\\ \hline
three non-$1$'s, otherwise& $\beta_C=3$
\\\hline
\end{tabular}
\end{center}
\caption{\label{tab:3x3-ones} Complete characterization of $3\times3$ games \eqref{eq:general3x3} with respect to the contradiction number if some of $x,y,z,w$
are equal to $1$.}
\end{table}

The five conditions in the second to last row of Table \ref{tab:3x3} can also be expressed by the equalities:
\begin{enumerate}[i)]

\item $w^{-1}yz^{-1}x=1$ $\Leftrightarrow$ $yw^{-1}=zx^{-1}$;
\item $yz^{-1}x=1$ $\Leftrightarrow$ $zx^{-1}=y$;
\item $w^{-1}z^{-1}x=1$ $\Leftrightarrow$ $xw^{-1}=z$;
\item $w^{-1}yx=1$ $\Leftrightarrow$ $wy^{-1}=x$;
\item $w^{-1}yz^{-1}=1$ $\Leftrightarrow$ $yz^{-1}=w$.
\end{enumerate}

In fact, the matrix can be transformed as follows.

\vspace{2ex}
\begin{align}
&\left(\begin{array}{ccc}
1 & 1 & 1\\ 
1 & w & x\\
1 & y & z
\end{array}\right) \rightarrow 
\left(\begin{array}{ccc}
1 & w^{-1} & x^{-1}\\ 
1 & 1 & 1\\
1 & yw^{-1} & zx^{-1}
\end{array}\right) \rightarrow 
\nonumber \\
&\rightarrow \left(\begin{array}{ccc}
1 & w^{-1} & x^{-1}\\ 
1 & 1 & 1\\
z^{-1}x & yw^{-1}z^{-1}x & 1
\end{array}\right)    
\end{align}
or
\begin{align}
&\left(\begin{array}{ccc}
1 & 1 & 1\\ 
1 & w & x\\
1 & y & z
\end{array}\right) \rightarrow 
\left(\begin{array}{ccc}
1 & 1 & x^{-1}\\ 
1 & w & 1\\
1 & y & zx^{-1}
\end{array}\right) \rightarrow 
\nonumber \\
&\rightarrow \left(\begin{array}{ccc}
1 & 1 & x^{-1}\\ 
1 & w & 1\\
z^{-1}x & yz^{-1}x & 1
\end{array}\right)    
\end{align}
\vspace{2ex}

\blk 
Assuming all of the values $\{w,x,y,z\}$ are different, it is clear that the final matrices in the above transformations have three contradictions iff the values satisfy conditions i) and ii), respectively. Conditions iii) - v) can be analyzed via similar transformations.

\medskip Finally, lets us comment on the case when some of the parameters $w,x,y,z$ in \eqref{eq:general3x3} are equal to $1$. The analysis of most of such instances is implicit in the argument above.  For example, if there is only one entry that is different from $1$, then $\beta_C(M)=1$. If there are two such entries contained in a single row or column, then
$\beta_C$ equals $1$ or $2$ depending on whether these entries are equal or not (Proposition \ref{prop:one_row}). If the two entries different from $1$ do not belong to the same row nor column, we have an instance of a matrix of type \eqref{2elements}, which -- as we determined -- has the contradiction number equal to $2$. If exactly three entries are different from $1$, say 
\begin{align}
\label{eq:3elements}
M=\left(\begin{array}{ccc}
1 & 1 & 1\\ 
1 & 1 & x\\
1 & y & z
\end{array}\right),     
\end{align}
then -- as in the proof of Proposition \ref{prop:diag3by3} -- a necessary condition for $\beta_C(M)\leq 2$ is that at least one $2\times 2$ minor  of $M$ is zero in addition to the obvious one ($i,j \in \{0,1\}$). By direct checking, we see that this happens only if $y=z$ or $x=z$, or if $xy=z$. It can be easily verified that all these conditions are also sufficient for $\beta_C(M)= 2$. 
\section{Summary}

{\cred This paper explores nonlocal games, with a focus on attempts to construct explicit examples of linear games that exhibit a significant gap between classical and quantum values. The main objective is to analyze classical values of these games using tools from graph theory and number theory, ultimately developing a novel technique, the girth method, which allows for the construction of games with low classical values.
The \emph{girth method} combines results from graph theory, particularly regarding the girth of a graph (the length of its shortest cycle), and number theory and harmonic analysis. This approach enables constructing games with minimal classical values and potential for a large quantum-to-classical value ratio. The paper provides explicit examples of such games, showing that they have an unbounded quantum violation when comparing their classical values with known upper bounds for quantum values.
Additionally, the connection between nonlocal games and graph theory is emphasized. The paper shows that the structure of a graph, particularly the number of edges and cycles, is directly related to the classical value of the corresponding game. Games whose graphs have a small number of edges and a high girth tend to have low classical value, enhancing the potential for quantum violation.

}

\vspace{3mm}

{\it Acknowledgments}
We thank Ravishankar Ramanathan for discussions, and for pointing out the 
notion of matrix rigidity. We also thank Mateus Araujo for comments and for pointing out a mistake in previous version regarding maximal number of contradictions for $6\times 6$ matrix. We express our gratitude to the anonymous reviewers for their insightful suggestions.
The research of SJS was supported in part by the grant DMS-1600124 from the National Science Foundation (USA). 
MR, AR, PG and MH are supported by National Science Centre, Poland,
grant OPUS 9. 2015/17/B/ST2/01945.
M.H. also acknowledges support from the Foundation for Polish Science through IRAP project co-financed by EU within the Smart Growth Operational Programme (contract no.2018/MAB/5). 



\bibliographystyle{unsrtnat}

\clearpage
\newpage

\section*{Appendix }

\begin{proof}[Proof of Proposition \ref{prop1}]
Let 
\begin{align}
M=\left(\begin{array}{cccc}
m_{00} & m_{01} & ... & m_{0k}\\ 
m_{10} & m_{11} & ... & m_{1k}\\
\vdots & \vdots &     & \vdots\\
m_{l0} & m_{l1} & ... & m_{lk}
\end{array}\right).    
\end{align}
After multiplying each row by $m_{i0}^{-1},$ where $i$ is the number of the row, we obtain the matrix
\begin{align}
\left(\begin{array}{cccc}
1 & m_{01}m_{00}^{-1} & ... & m_{0k}m_{00}^{-1}\\ 
1 & m_{11}m_{10}^{-1} & ... & m_{1k}m_{10}^{-1}\\
\vdots & \vdots &     & \vdots\\
1 & m_{l1}m_{l0}^{-1} & ... & m_{lk}m_{l0}^{-1}
\end{array}\right).    
\end{align}
Next we multiply columns $1$ - $k$ by $m_{00}m_{0j}^{-1},$ where $j$ is the number of the column to obtain
\begin{align}
M'=\left(\begin{array}{cccc}
1 & 1  & ... & 1\\ 
1 & m_{11}m_{10}^{-1}m_{00}m_{01}^{-1} & ... & m_{1k}m_{10}^{-1}m_{00}m_{0k}^{-1}\\
\vdots & \vdots &     & \vdots\\
1 & m_{l1}m_{l0}^{-1}m_{00}m_{01}^{-1} & ... & m_{lk}m_{l0}^{-1}m_{00}m_{0k}^{-1}
\end{array}\right).    
\end{align}

\noindent $M'$ is a matrix in which all elements of the first row and the first column are equal to $1$ {and it is  equivalent to $M$ by construction, as needed.} 
\end{proof}

{Note that essentially the same argument allows to obtain $M'$ in which the locations of entries equal to $1$ correspond to any given tree 
represented as a bipartite graph on  $n_A\times n_B$  vertices. } 


\medskip
\begin{proof}[Proof of Lemma \ref{lll}]
Let $M$ be a matrix of the form \ref{eq:ones}. If two non-one elements in the same row (or column) are equal, we can transform the matrix as follows
\begin{align}
&\left(\begin{array}{cccccc}
1      & 1      & 1      &  1     & ... & 1\\ 
1      & x      & x      & m_{13} & ... & m_{1l}\\
1      & m_{21} &        & ...    &     & m_{2l}\\
\vdots & \vdots &        &        &     & \vdots\\
1      & m_{k1} &        & ...    &     & m_{kl}
\end{array}\right) \rightarrow  \nonumber \\
&\rightarrow \left(\begin{array}{cccccc}
1      & 1      & 1      &  1     & ... & 1\\ 
m^{-1} & 1      & 1      & m_{13}m^{-1} & ... & m_{1l}m^{-1}\\
1      & m_{21} &        & ...    &     & m_{2l}\\
\vdots & \vdots &        &        &     & \vdots\\
1      & m_{k1} &        & ...    &     & m_{kl}
\end{array}\right).
\end{align}
\vspace{2ex}
If the two equal elements are neither in the same row, nor in the same column, we have

\begin{align}
&\left(\begin{array}{cccccc}
1      & 1      & 1      &  1     & ... & 1\\ 
1      & m_{11} & m      & m_{13} & ... & m_{1l}\\
1      & m      & m_{22} & ...    &     & m_{2l}\\
\vdots & \vdots &        &        &     & \vdots\\
1      & m_{k1} & m_{k2} & ...    &     & m_{kl}
\end{array}\right) \rightarrow 
\nonumber \\
&\rightarrow \left(\begin{array}{cccccc}
1      & m^{-1}       & m^{-1}      & m_{13}^{-1}  & ... & m_{1l}^{-1}\\ 
1      & m_{11}m^{-1} & 1           & 1            & ... & 1      \\
1      & 1            & m_{22}m^{-1}& ...          &     & m_{2l}m_{1l}^{-1}\\
\vdots & \vdots       &             &              &     & \vdots\\
1      & m_{k1}m^{-1} & m_{k2}m^{-1}& ...          &     & m_{kl}m_{1l}^{-1}
\end{array}\right)\rightarrow
\nonumber \\
&\rightarrow \left(\begin{array}{cccccc}
x           & 1            & 1           & mm_{13}^{-1} & ... & mm_{1l}^{-1}\\ 
1           & m_{11}m^{-1} & 1           & 1            & ... & 1      \\
1           & 1            & m_{22}m^{-1}& ...          &     & m_{2l}m_{1l}^{-1}\\
\vdots      & \vdots       &             &              &     & \vdots\\
1           & m_{k1}m^{-1} & m_{k2}m^{-1}& ...          &     & m_{kl}m_{1l}^{-1}
\end{array}\right)
\end{align}.

\vspace{2ex}
In both cases the contradiction number is shown to be less than the maximum. 
\end{proof}

However, simply making all elements different from 1 distinct is not enough to ensure the maximum number of contradictions.

\medskip 
\begin{proof}[Proof of Proposition \ref{prop:chromatic}]
The maximum number of contradictions in an $n\times n$ matrix is only achieved if every cycle in the corresponding graph $K_{n,n}$ contains a contradiction. Thus for every pair of permutations $\pi_1,\pi_2\in S_n$ which defines exactly one cycle the sums
$s_i=\sum\limits_{j=0}^{n-1} k_{j,\pi_i(j)}$ 
must be different. Therefore, assigning these sums to the vertices of $G_n$ produces a proper coloring. Since $s_i$ can have no more than $d$ different values, this is impossible for $d<\chi(G_n).$

Notice that for $N_{G_n}(\id)$ is the set of all cyclic permutations $\pi=(x_1\ldots x_t)$. It is easy to see that if $\pi_1\pi_2 = (x_1\ldots x_t)$ is a cycle then $\pi_1\sigma\pi_2\sigma = (y_1...y_t)$ is also a cycle for any permutation $\sigma\in S_n$. It follows that $\pi_1,\pi_2$ are adjacent in $G_n$ if and only if $\pi_1\pi_2^{-1}=(x_1...x_t)$ and that every function $f:S_n\mapsto S_n$, where $f(\pi)=\pi\sigma$ is an automorphism on $G_n.$ Thus, the largest set of cyclic permutations $C\subset S_n$ such that $\pi_i\pi_j^{-1}$ is a cycle for any $\pi_i, \pi_j\in C$, plus $\id$, is the largest clique in $G_n.$ It is well known that for any graph $G$ the chromatic number is at least the size of the largest clique in $G$. 

The independence number of $G_n$ is the size of the largest set $J$ of cyclic permutations such that $\pi=\pi_i\pi_j^{-1}$ is not a cyclic permutation for any $\pi_i,\pi_j\in J.$ Since $\chi(G)\geq \frac{\left|V(G)\right|}{\alpha(G)}$ for any $G$, we have $\chi(G_n)\geq\frac{n!}{\left|J\right|}.$
\end{proof}






\end{document}